\documentclass[12pt,reqno]{amsart}
\usepackage[round]{natbib}
\usepackage{amsfonts}
\usepackage{pdfsync}
\usepackage{amssymb, amsmath}
 \usepackage{graphicx,graphics,xcolor,epsfig,setspace}
\usepackage[section]{placeins}
\def\sqr#1#2{{\vcenter{\vbox{\hrule height.#2pt
        \hbox{\vrule width.#2pt height#1pt \kern#2pt
        \vrule width.#2pt}
        \hrule height.#2pt}}}}

\newcommand{\nc}{\newcommand}
\nc{\parent}[1]{$[\![#1]\!]$}
\newtheorem{theorem}{Theorem}[section]
\newtheorem{lemma}{Lemma}[section]

\newtheorem{proposition}{Proposition}[section]
\newtheorem{remark}{Remark}
\newtheorem{definition}{Definition}[section]

\DeclareMathOperator*{\argmin}{arg\,min}
\DeclareMathOperator*{\argmax}{arg\,max}
\nc{\cadlag}{c\`{a}dl\`{a}g } \nc{\ba}{\begin{array}}
\nc{\ea}{\end{array}} \nc{\be}{\begin{equation}}
\nc{\ee}{\end{equation}} \nc{\bea}{\begin{eqnarray}}
\nc{\eea}{\end{eqnarray}} \nc{\bean}{\begin{eqnarray*}}
\nc{\eean}{\end{eqnarray*}} \nc{\bu}{\bullet} \nc{\nn}{\nonumber}
\nc{\cA}{{\mathcal A}} \nc{\cB}{{\mathcal B}} \nc{\cC}{{\mathcal C}}\nc{\cX}{{\mathcal X}}
\nc{\cD}{{\mathcal D}} \nc{\bbD}{\mathbb{D}}\nc{\bbH}{\mathbb{H}}
\nc{\bbF}{\mathbb{F}}\nc{\bbG}{\mathbb{G}}\nc{\cG}{{\mathcal G}} \nc{\cF}{{\mathcal F}}
\nc{\cS}{{\mathcal S}} \nc{\cU}{{\mathcal U}} \nc{\cH}{{\mathcal H}}
\nc{\cK}{{\mathcal K}} \nc{\cL}{{\mathcal L}} \nc{\cM}{{\mathcal M}}
\nc{\cO}{{\mathcal O}} \nc{\cP}{{\mathcal P}} \nc{\cQ}{{\mathcal Q}}\nc{\bbE}{\mathbb{E}}
\nc{\bbEQ}{\mathbb{E}_{\mathbb{Q}}} \nc{\eps}{\varepsilon}
\nc{\bbEP}{\mathbb{E}_{\mathbb{P}}}\nc{\bbL}{\mathbb{L}}
\nc{\bbP}{\mathbb{P}} \nc{\bbQ}{\mathbb{Q}} \nc{\del}{\partial}
\nc{\Om}{\Omega} \nc{\om}{\omega} \nc{\bbR}{\mathbb{R}}
\nc{\bbC}{\mathbb{C}} \nc{\bfr}{\begin{flushright}}
\nc{\efr}{\end{flushright}} \nc{\dXt}{\Delta X_{t}} \nc{\dXs}{\Delta
X_{s}} \nc{\bs}{\blacksquare} \nc{\dX}{\Delta X} \nc{\dY}{\Delta Y}
\nc{\dnkx}{\left(X(T^{n}_{k})-X(T^{n}_{k-1})\right)}
\nc{\esssup}{\mathrm{ess}\mbox{ }\mathrm{sup}}
\nc{\essinf}{\mathrm{ess}\mbox{ } \mathrm{inf}}
\nc{\dhats}{\widehat{\delta_s}} \nc{\tX}{\tilde{X}}
\nc{\tZ}{\tilde{Z}}
\nc{\what}{\widehat}
 \nc{\half}{\frac{1}{2}}

\def\rar{\rightarrow}

\nc{\chf}{\mbox{$\mathbf1$}} \nc{\eid}{\stackrel{d}{=}}

\begin{document}
\title{Order routing and market quality: Who benefits from internalization?}
\author{Umut \c{C}et\.in}
\address{Department of Statistics, London School of Economics and Political Science, 10 Houghton St., London, WC2A 2AE}
\email{u.cetin@lse.ac.uk}
\author{Albina Danilova}
\address{Department of Mathematics, London School of Economics and Political Science, 10 Houghton St., London, WC2A 2AE}
\email{a.danilova@lse.ac.uk}
\date{\today}
\begin{abstract}
	We analyse two models of liquidity provision to determine the retail traders' preference for marketable order routing. Order internalization is captured by a model of market makers competing for the retail order flow in a Bertrand fashion. On the other hand, the price-taking competitive liquidity providers characterize the open exchange model.  We show that, when liquidity providers are risk averse, routing of the marketable orders to the wholesalers is preferred by {\em all} retail traders: informed, uninformed and noisy.  The unwillingness of liquidity providers to bear risk causes the strategic trader (informed or not) to absorb large shocks in their inventories.  This results in mean reverting inventories, price reversal, and lower  market depth.  The equilibria in both models coincide with \cite{Kyle} when liquidity providers are risk neutral. We also identify a universal parameter that allows comparison of market liquidity, profit and value of information across  different markets. 
\end{abstract}
\maketitle

\section{Introduction}

{The retail internalization, i.e. the practice of routing the marketable retail orders to wholesalers such as Citadel, Susquehanna, and Wolverine,  has been a major focus for regulators since the late 1990s. The size of the market and the proportion of retail orders affected by this practice completely justify  the regulators' concern.  Indeed,  \cite{BPS} report that  the average share of internalized trades in the total weekly stock trading volume  is 17\%, and has an upward trend.  Moreover, most of the retail orders received by brokers are routed to wholesalers\footnote{The US Securities and Exchange Commission (2010) Concept Release on Equity Market Structure states: ``A review of the order routing disclosures required by Rule 606 of Regulation NMS of eight broker-dealers with significant retail customer accounts reveals that nearly 100\% of their customer market orders are routed to OTC market makers.'' The concept release estimates the amount of payment for order flow is 0.1 cent per share or less.
	https://www.sec.gov/rules/concept/2010/34-61358.pdf (2010)}.

The  regulatory debate on internalization is far from over: while UK Financial Conduct Authority (FCA) has banned payment for order flow (PFOF) since 2012 and the European Union is following the suit in their 2021 Capital Markets Union Package\footnote{See https://ec.europa.eu/commission/presscorner/detail/en/ip\_21\_6251},  the US Securities and Exchange Commission (SEC) is still assessing the impact of such a ban on market quality\footnote{In their earlier, 1997,  assessment  SEC did not find internalization harmful to market quality.}. 

The  main argument for banning PFOF is that it reduces liquidity on exchanges since the internalization of {\em uninformed} retail orders increases the information asymmetry on the exchanges leading to higher spreads. This in turn  increases their execution costs as the price improvement that the wholesalers must provide is based on the best quotes from the exchange\footnote{See the CFA Institute study at https://www.cfainstitute.org/en/advocacy/policy-positions/payment-for-order-flow-in-the-united-kingdom.}. \cite{EKO96} and \cite{BK97} similarly argue that PFOF lowers the market quality since the wholesalers {\em cream skim} the uninformed traders. 

This argument hinges on the assumption that the retail orders are uninformed, consistently with empirical studies  at that time (see, e.g. \cite{BO00,BO08}. However, more recent studies (\cite{KST08}, \cite{BOZ08}, \cite{KLST12}, \cite{KT13}, \cite{FGL14}, \cite{BKS16}, \cite{BJZZ21})  show that the retail order flow can predict future returns, which suggests it can contain informed trading. 

Our paper enters the debate  on optimal routing of the retail order flow by comparing two  continuous time models for liquidity provision: in the first one, that  extends  \cite{GP},  the liquidity is supplied by perfectly  competitive agents of unit mass (as a model for the lit exchange); in the second, that extends  \cite{Kyle}, it is provided by a finite number of imperfectly competitive market makers, i.e. wholesalers.  Considering both informed and uninformed types of strategic traders allows us to analyse the retail agents' preferences for trading venues whether the retail order is informed or uninformed. We find that independently of the information content of the retail order strategic traders prefer their orders to be routed to the wholesalers when the liquidity providers are risk averse.  Remarkably this profit  gain  of strategic traders doesn't come at the expense of noise traders who are indifferent between two venues.  This result stems from the fact that the competition drives the market makers' utility gain to zero whereas the competitive agents enjoy positive utility gain in equilibrium.  

These findings are consistent with the proponents of internalization who argue that it enhances market quality by increasing the number of agents competing to execute order flow. For example, \cite{HNV98}  show that internalized trades pay lower spreads than non-internalized ones and that bid-ask spreads are not affected by  the level of internalization.   \cite{Bat97} and \cite{BGJ97} demonstrate, correspondingly, that the transaction costs are not increased by the  introduction of PFOF and internalizing dealers.  

This paper also contributes to theoretical market microstructure literature in two ways.  First, the market structure with competitive agents extends   \cite{GP} to allow adverse selection. Second, the model with market makers considered in this paper  extends the continuous time version of the model introduced in \cite{Kyle} by allowing the market makers to be risk averse and by providing an explicit description of their competition.  Endowing the liquidity providers with  CARA utility  is not only consistent with  empirical research\footnote{ The literature indicates that the trading behaviour of the market makers is affected by their inventory positions even after isolating the adverse selection component of their quotes (see \cite{HS97} and \cite{MS93} for NYSE, \cite{HNV98} for LSE, \cite{BR05} for FX; for a survey of related literature and results, see Sections 1.2 and 1.3 in \cite{BGSsurvey}). This  leads to the conclusion that the market makers are risk averse as inventory should play no role for a risk neutral market maker once the adverse selection is accounted for.},  but also allows to assess the impact of risk aversion on market liquidity.

The risk aversion of liquidity providers alters the equilibrium outcome drastically. One of the main differences is that, in contrast with the earlier extensions of the Kyle model, the insider's trades induce a reversal in the total demand (and hence { in} the price) { in both models of liquidity provision}  process, i.e. { insider trading} ceases to be {\em inconspicuous}.  This is a remarkable departure from the results of \cite{Kyle} and related literature, where the distribution of the total demand is unaltered by the presence of private information, and is fully determined by the noise trading. Inventories of liquidity suppliers are mean reverting in equilibrium  irrespectively  of the nature of liquidity provision and independently of whether the strategic trader is informed or not.  The driving force behind this result is that the risk aversion of liquidity providers makes them unwilling to bear risk. Instead of paying the extra compensation for the risk, the strategic trader chooses to reduce it by absorbing the part of  large fluctuation in the demand from the noise traders. This causes the total demand to mean revert -- a result unanimously supported by empirical studies (see, among others, \cite{HS97}, \cite{MS93}, \cite{HNV98}, and \cite{BR05}).
 Moreover, the fact that the mean reversion is more pronounced in the presence of private information for reasonable market parameters suggests a new paradigm of empirical research which has so far attributed the mean reversion solely to the inventory costs.
 
In addition to studying equilibrium outcomes resulting from different models of liquidity provision and risk aversion of agents, we are also able to compare the market depth and value of information across different models and markets due to identification of a universal parameter, {\it the market adjusted risk aversion} $\rho_M$. Remarkably, the market outcomes, once normalized with their counterpart in Kyle's model, turned out  be fully determined by $\rho_M=\rho\sigma\gamma$ rather then by individual market characteristics such as risk aversion $\rho$, volatility of the noisy trades, $\sigma$, and volatility of the asset, $\gamma$, in isolation.

We observe that the risk aversion considerably reduces the market depth as expected. Moreover,  higher $\rho_M$ leads to  lower market depth since the liquidity providers require additional compensation for the risk that they bear.  The same mechanism is responsible for the monotone relationship between normalized market depth and market adjusted risk aversion.

A surprising outcome of the model, however, is that the value of private information  is non-monotone  in $\rho_M$ in both markets. In particular we observe that it is decreasing for small market adjusted  risk aversion and increasing for sufficiently large values. The reason for this apparently counter-intuitive behaviour is the fact that the value of information consists of two components. One component reflects insider's profits purely due to her private information, while the other is due to her participation in risk sharing. The mean reverting demand causes the former to decrease in the market adjusted  risk aversion, since this profit is collected by bringing the price to its fundamental value, and the stronger is the mean reversion, the more effort is needed to do so.  On the other hand, the latter must increase in the market adjusted  risk aversion as a result of risk sharing. These two components combined result in a non-monotone dependency of the value of information on market adjusted risk aversion: for small $\rho_M$ the value of information purely due to private information dominates, whereas for higher $\rho_M$ this value is largely determined by profits due to market's response.


These observations show that a mere introduction of risk averse market makers to the setting of \cite{Kyle} changes the equilibrium outcome fundamentally. The previous attempt to investigate the effect of such an extension was carried out in \cite{subRA}, who considered a one-period model. Due to the strategic behaviour of the insider, a direct extension of this one-period model to a multi-period, let alone a continuous-time one, is not possible, as intuited in \cite{subRA}. As we demonstrate in this paper one can, however, derive the closed-form equilibrium in a continuous-time setting in the spirit of \cite{Back92}.

Related work contains two streams of literature.  In the first stream that stems from the extension of the Kyle's model liquidity providers are market makers who compete in a Bertrand fashion and set the quotes.  Extensions of the Kyle's model include, but are not limited to, continuous-time formulation by \cite{Back92},  risk-averse insider by \cite{Bar02},  long-lived private information by \cite{BP98},  markets with default risk by \cite{CCDdef}, competition among insiders by \cite{BCWmult}, and stochastic volatility of noise trades by \cite{CDF16}.  In all these models the optimal strategy of the insider is inconspicuous. In the second strand  of literature that stems from   \cite{GP} the liquidity providers are price-taking dealers that correspond our competitive agents in Section \ref{s:LP}. \cite{BCEL20} is the work in this stream of literature  that is closest to  the model studied in Section \ref{s:LP}.  \cite{BCEL20} recasts \cite{GP} as the insider trading model with risk-averse dealers by employing optimal transport and formulating the problem under the risk-neutral measure rather than physical as we do.  As their focus is primarily on the solvability of this extension, they do not consider the liquidity providers that set the quotes, i.e. market maker, precluding the analysis of the choice of the trading venue that is the main focus of our paper. 

The rest of the paper is structured as follows. In Section \ref{model} we describe the market structure. Section \ref{s:LP} establishes the existence of equilibrium when liquidity is provided by  competitive agents and characterises it. Section \ref{s:MM} derives an equilibrium in the market makers' case. Section \ref{s:cd} discusses market parameters resulting from equilibria derived in Sections \ref{s:LP} and \ref{s:MM}  and their relationship with the Kyle model. Section \ref{s:C} concludes.

}

\section{Market structure} \label{model}
{Let $(\Omega , \cF , (\cF_t)_{t \in [0,1]} , \bbP)$ be a filtered probability space  satisfying the usual conditions of right continuity and $\bbP$-completeness. Assume that on this probability space there exist a normal random variable $V\in \cF_0$ with mean $\mu$ and variance $\gamma^2$-- the fundamental value of the asset,  and a standard Brownian  motion $B$, independent of $V$. Since all information shocks in the model will come from $V$ and $B$, we can and will assume that $\cF$ is generated by $V$ and $(B_t)_{t=0}^1$.  We also assume that all filtrations in this paper satisfy the usual conditions.

We consider a market in which a single risky asset is traded. The value of this asset, $V$, will be  public knowledge at some future time $t=1$. For simplicity of exposition, we assume that the risk free interest rate is $0$.

There are three types of agents that interact in this market:
\begin{itemize}
\item[i)] Noise traders, whose demands are random, price inelastic, and do not reveal any information about the value of $V$. In particular, we assume that their cumulative demand at time $t$ is given by $Z_t$ -- a Brownian motion with mean $0$ and variance $\sigma^2$, independent of $V$. That is, $Z_t=\sigma B_t$, where $B$ is a standard Brownian motion independent of $V$.
\item[ii)] A single insider, who knows $V$ from time $t=0$ onwards, and is risk neutral. We will denote insider's cumulative demand at time $t$ by $X_t$. The filtration of the insider, $\cF^I$, is generated by observing the price of the risky asset and $V$.
\item[iii)] Liquidity suppliers observe only the net demand\footnote{That is, we assume that the traders cannot choose a particular liquidity supplier to trade with. Hence, the demand is aggregated before being sent to the liquidity suppliers.} of the risky asset, $Y=X+Z$, thus, their filtration, $\cF^M$, is generated by $Y$.

    We also assume that they have identical CARA utilities $U$ with the common risk aversion parameter $\rho$. More precisely, $U(x)=1-e^{-\rho x}$. We will consider two markets where the liquidity suppliers are either
    \begin{itemize}
    	\item[a)] {\em perfectly competitive agents} that form a continuum of mass one and take prices $P$ as given, or
    	\item[b)] {\em market makers} who compete in a Bertrand fashion for the net demand of the risky asset. The number of market makers is assumed to be finite.
    \end{itemize}
\end{itemize}
The asymmetry of information between the insider and liquidity suppliers also leads to the conclusion that they have different probability measures. Indeed, since the insider is given the realization of $V=v$, she assigns probability $1$ to this event. On the other hand, the liquidity suppliers do not observe $V$, thus they assign probability $0$ to the event $V=v$ since $V$ is a continuous random variable. Thus, their probability measures are not only different but {\em singular}. We will denote by $\bbP^v$ the probability measure of the insider who is given the information that $V=v$. Note that $\bbP^v$ is the joint distribution of $V$ and $B$ evaluated at $V=v$, and the liquidity suppliers' probability measure $\bbP$ satisfies
\[
\bbP(E)=\int_{\bbR}\bbP^v(E)\bbP(V \in dv)
\]
for any $E \in \cF$.}
\section{Competitive agents equilibrium}\label{s:LP}
{Before defining equilibrium in this market we will define the admissible strategies for the competitive agents and the insider. 
We will look for an equilibrium where the price is of the following form:
\be \label{e:LPEprice}
dP_t^X=m_tdt +\lambda_t(dZ_t+dX_t),
\ee
for some appropriately measurable $m,\lambda$ and $X$ such that $P^X$ is $(\bbP,\cF^M)$-semimartingale. We will denote the class of $\cF^M$-predictable processes that are $P^X$-integrable by $\cL(P^X)$ (see Section IV.2 of \cite{Pro}). Moreover, since equilibrium price processes should not allow for arbitrage opportunities, we will only consider process $P^X$ admitting an equivalent local martingale measure; that is, there exists a  $\bbQ\sim \bbP$ under which $P^X$ is an $\cF^M$-local martingale. The processes satisfying these conditions are said to be {\em compatible with no-arbitrage} and the set of such processes is denoted by $\cP(X)$. 
\begin{definition}
	\label{d:lp:LPadmissible} Let $m$, $\sigma$ and $X$ be given and consider $P^X\in \cP(X)$. For each competitive agent $a\in[0,1]$ $\theta$ is admissible given $P^X$ if  $(\int_0^t \theta_s dP^X_s)_{t \in [0,1]}$ is a $(\bbQ,\cF^M)$-martingale, where $\bbQ$ is the unique equivalent local martingale measure for $P^X$. 
	The set of admissible strategies for competitive agents is denoted by $\cA(m,\lambda,X)$.
\end{definition}
The set of admissible strategies for the insider will be different as she takes into account the feedback effect of her trading.
\begin{definition}
	\label{d:lp:insadm}   Let $m$ and $\lambda$  be given.  The trading strategy $X$ is admissible if it is absolutely continuous, the associated price process is compatible with no-arbitrage, that is, 
	\[
	P_t := P_0+\int_0^t m_u du +\int_0^t \lambda_u(dZ_u +dX_u) \in \cP(X),
	\]
	 and no doubling strategies are allowed, i.e.
	\[
	\bbE^v\int_0^1P_t^2dt<\infty.
	\]
	
	The set of admissible strategies for the insider is denoted by $\cA^I(m,\lambda)$.
\end{definition}

In addition to the standard {\em no-doubling} condition as in \cite{Back92} we also require that the  price process associated with the insider's trading strategy should not offer arbitrage opportunities to the competitive agents since, otherwise, the model will degenerate.
\begin{definition}
	\label{d:lp:equilibrium} An equilibrium in this market is given by $((m^*,\lambda^*),\theta^*=(\theta^{*,a})_{a\in [0,1]},X^*)$ such that 
	\[
	P^*_u=P^*_0 +\int_0^u m^*_tdt +\int_0^u\lambda^*_t(dZ_t+dX^*_t),
	\]
	$P^* \in \cP(X^*)$, $\theta^{*,a}\in \cA(m^*,\lambda^*,X^*)$ for all $a \in [0,1]$, $X^*\in \cA^I(m^*,\lambda^*)$ and the following are satisfied:
	\begin{enumerate}
		\item {\bf (Optimality for competitive agents)} For all $a\in [0,1]$ $\theta^{*,a}$ solves
		\[
		\sup_{\theta \in \cA(m^*,\lambda^*,X^*)}\bbE\left[1-\exp\left\{-\rho\left(\int_0^1\theta_tdP^*_t +\theta_1(V-P^*_1)\right)\right\} \right];
			\]
		\item {\bf (Optimality for the insider)} $X^{*}$ solves
		\[
		\sup_{X\in \cA^I(m^*,\lambda^*)}\bbE^v\left[\int_0^1X_tdP_t +X_1(V-P_1) \right],
		\]
		where 
		\[
		P_t = P^*_0+\int_0^t m^*_u du +\int_0^t \lambda^*_u(dZ_u +dX_u);
		\]
		\item {\bf (Market clearing)} $\int_0^1 \theta^{*,a}_tda=-(X^*_t+Z_t)$ for all $t \in [0,1]$.
	\end{enumerate}
\end{definition}
This definition of equilibrium formalizes a continuous-time market in which Kyle's insider trades with  ``dealers'' of \cite{GP}. To be more precise this notion of equilibrium  postulates that the demand pressure faced by the dealers comprises of the optimal strategy of the insider and the inelastic demand of the noise traders. By doing so it partially endogenizes the total order flow that was fully exogenous in  \cite{GP}. As in Kyle, the insider takes into account the fact that her actions influence the prices whereas the liquidity suppliers are price takers as in \cite{GP}. 

In case of risk-neutral competitive agents the price process is a martingale and, therefore, the equilibrium will coincide with that of \cite{Kyle} in which the price is an affine function of demand. The next theorem establishes the existence of an equilibrium where the price process is still an affine function of demand for an arbitrary risk aversion parameter.
\begin{theorem}  \label{t:eqLP}
	There exists an equilibrium $((m^*,\lambda^*),\theta^*=(\theta^{*,a})_{a\in [0,1]},X^*)$, where
	\bean
	m^*&=&0,\\
	\lambda^*&=&\frac{\rho\gamma^2}{2}+\sqrt{\frac{\rho^2\gamma^4}{4}+\frac{\gamma^2}{\sigma^2}},\\
	P^*_0&=&\mu,\\
	dX^*_t&=&\frac{\frac{V-\mu}{\lambda^*}-(X^*_t+Z_t)} {1-t}dt, \quad \mbox{ and }\\
	\theta^{*,a}&=&-(X^*_t+Z_t), \quad \forall a\in [0,1].
	\eean
Moreover, the equilibrium characteristics of the market are as follows:
\begin{enumerate}
	\item From the point of view of the insider the equilibrium demand and price dynamics are given by
	\bea 
	dY^*_t&=&\sigma dB_t+\frac{\frac{V-\mu}{\lambda^*}-Y^*_t} {1-t}dt,  \label{e:demandLP}\\
	dP^*_t&=&\lambda^*\sigma dB_t +\frac{V-P^*_t}{1-t}dt \label{e:priceLP},
	\eea
		and her ex-ante profit is given by
	\[
	\bbE\left[\frac{(P_0-V)}{2\lambda^*}+\frac{\sigma^2}{2}\lambda^*\right]=\frac{\gamma^2}{2\lambda^*}+\frac{\sigma^2}{2}\lambda^*.
	\]
	\item From the point of view of the competitive agents the equilibrium demand and price dynamics are given by
	\bea 
	dY^*_t &=& \sigma d\beta^*_t -\sigma^2 \frac{\lambda^*\rho}{1+ \lambda^* \rho\sigma^2(1-t)}Y^*_t dt,
	\label{e:demandLPc}\\
	dP^*_t &=& \lambda^* \sigma d\beta^*_t -\sigma^2 \frac{\lambda^*\rho}{1+ \lambda^* \rho\sigma^2(1-t)}(P^*_t-\mu) dt \label{e:priceLPc},
	\eea
	where $\beta^*$ is an $\cF^M$-Brownian motion. Their expected utility  equals
	\be \label{e:lputility}
	1-\frac{\lambda^*\sigma}{\gamma}e^{-\frac{\lambda^*\rho\sigma^2}{2}}.
	\ee	
\end{enumerate}	
	\end{theorem}

The proof of the above theorem is technical and is delegated to the appendix. However, the idea of the proof is easy to grasp once we observe that  the proposed equilibrium price is a function of demand; thus,  the optimality condition for the insider's trading strategy is the same as in \cite{Back92}. That is, the insider drives the market price to the fundamental value $V$. Moreover, how she achieves that does not affect her utility and in particular the total demand does not have to be a Brownian motion as in \cite{Back92}.  Hence, to determine the optimal demand process that ensures convergence of the price to the fundamental value, one needs to solve the optimization problem of the competitive agents.  This is achieved by  equating the stochastic discount factor to their marginal utility, which allows us to identify  the law of the equilibrium demand from the point of view of the competitive agents. By combining this with Doob's h-transform technique we ensure that the price coincides with the fundamental value at the end of trading. 

In equilibrium the response of the insider to the choice of price as a function of total demand by the competitive agents is the same as in \cite{Back92}, and in particular does not depend on the level of the risk aversion of the competitive agents. On the other hand, the optimal strategy ceases to be inconspicuous in that the equlibrium demand from (\ref{e:demandLPc}) is an Ornstein-Uhlenbeck process as opposed to being a Brownian motion as in \cite{Back92}. 

The risk aversion of competitive agents introduces the inventory as another source of risk to be managed in comparison to \cite{Back92}.  For a given market depth the liquidity suppliers are thus willing to bear fluctuations in their inventory only to a certain extent at any given time. As the insider's strategy is optimal as long as the price converge to the fundamental value at the end, she opts for a strategy that keeps the inventory of competitive agents mean reverting around $0$. 

Keeping the liquidity suppliers' inventory mean reverting around $0$ is still optimal even if the strategic is uninformed as will be shown next.  The speed of mean reversion, however, is smaller in absence of private information.  This implies that although the mean reversion of demand is borne out of inventory management considerations, its speed is determined by the total amount of perceived risk, which is clearly higher in the presence of asymmetric information.
\begin{theorem}\label{t:benchLP} 
	There exists an equilibrium $((m^*,\lambda^*),\theta^*=(\theta^{*,a})_{a\in [0,1]},X^*)$, where
\bean
\lambda^*&=&\frac{\rho \gamma^2}{2},\\
P^*_t&=&\lambda^*Y^*_t +\mu, \quad t<1\\
P^*_1&=&2\lambda^*Y^*_1 +\mu,\\
	dX^*_t&=&-\frac{\lambda^*\sigma^2 \rho} {2+ \lambda^*\rho \sigma^2 (1-t)}(X^*_t+Z_t)dt, \quad t<1,\\
	\Delta X^*_1&=&-\frac{X^*_{1-}+Z_1}{2}, \quad \mbox{ and }\\
	\theta^{*,a}&=&-(X^*+Z), \quad \forall a\in [0,1].
\eean
Moreover, the ex-ante profit of the strategic trader is given by $\frac{\rho\gamma^2\sigma^2}{4}$ and the expected utility of the competitive agents equals
\be \label{e:lputilitys}
1-\sqrt{1+\frac{\lambda^*\rho\sigma^2}{2}}e^{-\frac{\lambda^*\rho\sigma^2}{2}}
\ee
\end{theorem}
}
\section{Market makers equilibrium} \label{s:MM}
{ Before determining the optimal behaviour of the agents in this market, we first need to understand how the orders are allocated among the market makers. Recall that the orders are combined before arriving to the market makers  and the Bertrand competition dictates that the total order is executed at the best available quote. To gain an intuition as to how the winning quote is determined, consider a small time interval $[s,s+ \Delta)$ and observe that the number of shares to be allocated is $\Delta Y:=Y_{s+\Delta}-Y_s$. Assume that $\Delta$ is such that $\Delta Y$ is arbitrarily small. Each market maker provides a quote $(\lambda^i, \phi^i)$, where $\lambda^i>0$. That is, the market maker $i$ is willing to absorb this demand at the price $\lambda^i \Delta Y+\phi^i$. 
	
	Observe that if the market makers quote different $\phi^i$s, the price at which traders buy an infinitesimal amount is smaller than the one at which they sell. This creates an opportunity for any trader to achieve infinite profits.  Indeed,  if $\Delta Y>0$, the order will be priced at $\min_{i}\phi^i +\lambda^{\argmin_i\phi^i}\Delta Y$ while the sell order is priced at $\max_{i}\phi^i +\lambda^{\argmax_i\phi^i}\Delta Y$ leading to a negative a negative bid-ask spread.  Thus, it is optimal for the market makers to quote the same $\phi$. 
	
	If  $\phi^i=\phi$ for all $i$, it is obvious that, independently from the sign of $\Delta Y$, the market maker quoting the smallest $\lambda^i$ provides the best price for this order.

 The order is allocated according to the {\em price-time priority}. That is, the market maker quoting the best price of the given order size gets the whole order. If there are several market makers quoting  the best price, the order is allocated to the one who has submitted this quote first.  Finally, as convention, the order is allocated to a single market maker at random  if several market makers has submitted the best price at the same time. We also rule out {\em shuffling}: The winning market maker continues to receive the total market order as long as he quotes the best price.

\subsection{Agents' objectives} The above discussion makes it clear that the market makers must quote the same $\phi$ at all times. Thus, the demand in any infinitesimal interval is allocated to the market maker quoting the smallest $\lambda$ and the market price evolves as
\be \label{e:trprice}
P_t=\phi(t) +\int_0^t \lambda(s)dY_s, \mbox{ where } \lambda(t) =\min_i \lambda^i(t), \, \forall t \in [0,1].
\ee 
As we are searching for a linear equilibrium we will consider the quotes of the form $(\lambda^i, \phi)$ that are admissible in the following sense.
\begin{definition}
The vector of quotes $(\Lambda, \Phi):=\left[ (\lambda^1, \phi^1) \ldots (\lambda^N, \phi^N)\right]$ submitted by the market makers is admissible if $\lambda^i>0$ and $\phi^i$ are piecewise constant and right-continuous functions of time such that $\phi^i=\phi$ for all $i=1, \ldots, N$. The set of admissible vectors is denoted by $\cH^M$. 
\end{definition}
Let $\cS$ be the set of right-continuous piecewise constant functions on $[0,1]$. Given an admissible vector of quotes, $(\Lambda, \Phi)$ we may define the function $H:\cH^M\mapsto \cS\times \cS$ by
\be \label{d:H}
H((\Lambda, \Phi)):=(\lambda, \phi),
\ee
where $\lambda$ is defined in (\ref{e:trprice}) and $\phi=\phi^1$. We will denote by $\cH$ the set of market quotes $(\lambda, \phi)$ that results from admissible vector of quotes submitted by the market makers. That is,
\[
\cH=\{(\lambda,\phi): (\lambda, \phi)=H((\Lambda, \Phi)), \; (\Lambda, \Phi)\in \cH^M\}.
\]
Note that for the given market quote $(\lambda,\phi)\in \cH$ the market price $P$ evolves as
\be \label{e:trpriceMQ}
P_t=\phi(t) +\int_0^t \lambda(s)dY_s, \, \forall t \in [0,1].
\ee 

Since the insider observes the market price, she observes the market quote  $(\lambda,\phi)$ as well as the demand process $Y$ as $\lambda>0$ for $(\lambda,\phi)\in \cH$. Therefore, the insider can perfectly infer the demand of the liquidity traders since she knows her own demand. This implies that   $\cF^I_t=\sigma(V,Z_s; s\leq t)$. Moreover, the fact that the insider can infer only the market quote rather than the vector of quotes submitted by the market makers implies that she makes her trading decision based on the market quote and, thus, the admissible strategies of the insider can be defined as follows:

\begin{definition} Given  $(\lambda,\phi)\in \cH$, $X$ is an admissible strategy for the insider if it satisfies the following.
	\begin{enumerate}
		\item $X$ is adapted to $\cF^I$.
		\item No doubling strategies are allowed:
		\be \label{a:Xintegrability}
		\bbE^v\int_0^1 P_t^2\,dt =\bbE^v\int_0^1\left(\int_0^t \lambda(s)\left\{dX_s+dZ_s\right\}+\phi(t)\right)^2 dt < \infty,
		\ee
		where  $\bbE^v$ is the expectation taken with respect to $\bbP^v$.
		\item $X$ is of finite variation.
	\end{enumerate}
A strategy satisfying the above conditions is called admissible for the given market quote $(\lambda, \phi)$ and the class of such strategies will be denoted by $\cA(\lambda, \phi)$. 
\end{definition}
 Condition 1) is the minimal requirement  that the strategy is implementable. Condition 2) prevents the insider from following doubling strategies as in \cite{Back92}. The last condition is a relaxation of the restriction to absolutely continuous strategies that is standard in the literature. This restriction is justified in \cite{Back92} under the assumption that $\lambda$ is continuous. Arguments similar to the ones used by Back can be employed to rule out jumps and martingale components on the intervals of continuity of $\lambda$. However, one cannot infer the sub-optimality of jumps at the points of discontinuity of $\lambda$ that leaves us with $X$ being of finite variation.
 
Observe that if $X \in \cA(\lambda, \phi)$, then the terminal wealth of the insider is given by\footnote{For a derivation and motivation of this equation see  \cite{Back92}.}
\be \label{d:iw}
W_1^X:=\int_0^1 X_{s-} dP_s + X_1(V-P_1)=\int_0^1\left(V-P_s\right)dX_s.
\ee
The first term in the terminal wealth corresponds to continuous trading in the risky asset, while the second term exists due to potential discontinuity in the asset price when the value becomes public knowledge at time $t=1$. The second expression for the wealth follows from integration by parts and $X$ being of finite variation.

As we assume that the insider is risk neutral, her optimisation problem consists of finding the strategy that maximises her expected profit for the given market quote. 

Since all  market makers are identical in this model, we will focus on the symmetric equilibria. To analyse the deviations from a symmetric vector of quotes consider $(\Lambda,\Phi)\in \cH^M$ where $(\lambda^i,\phi^i)=(\lambda,\phi)$ for $i\neq j$ and $(\lambda^j,\phi^j)=(\tilde{\lambda},\phi)$.  Before quantifying the profitability of such deviations we need to formalise the order allocation process for given total demand process $Y$. 

The deviating market maker gets the infinitesimal change in the demand at time $t>0$  if $\tilde{\lambda}(t)<\lambda(t)$. Let's denote by $\tau_j$ the first time that the market maker $j$ is not allocated the order. If $\tilde{\lambda}(0)>\lambda(0)$, $\tau_j$ is obviously $0$. If they are equal, $\tau_j$ is still $0$ if the market maker was not allocated the order at time $0$. In case he wins the allocation at time $0$,  this time is given by
\[
\tau^j=\left\{ 
\ba{ll}
\inf\{t\geq 0:  \tilde{\lambda}(t)>\lambda(t)\},& \mbox{ if } \tilde{\lambda}(0)=\lambda(0);\\
\inf\{t\geq 0:  \tilde{\lambda}(t)\geq \lambda(t)\},& \mbox{ if } \tilde{\lambda}(0)<\lambda(0),
\ea\right.
\]
where by convention the infimum of an empty set is $\infty$. Then,  according to our time-priority and no shuffling conventions, his holdings, $Y^j$, evolve as
\be \label{e:Yjdef}
dY^j_t = \left(\chf_{[\tau^j> t]}+\chf_{[\tau^j\leq t]}\chf_{[ \tilde{\lambda}(t)<\lambda(t)]}\right)dY_t.
\ee
To be more precise, the time-priority and no-shuffling imply that after losing the order allocation for the first time, the only way to receive  new market orders is via undercutting the market. 

We now turn to the characterisation of an equilibrium $((\Lambda,\Phi),X)$, where all market makers submit the same quote $(\lambda,\phi)$ and $X$ is an optimal strategy for $(\lambda,\phi)$. 

Recall that the market makers do not compete on $\phi$ and consider a situation where market maker $j$ deviates by submitting the quote $(\lambda^j,\phi)$. Let $\tilde{\lambda}:=\min\{\lambda, \lambda^j\}$ and note that the market quote  corresponding to this deviation is $(\tilde{\lambda},\phi)$.  First consider the case in which the deviation doesn't alter the market quote, i.e. $\tilde{\lambda}=\lambda$. Since the insider bases her decisions on the market quotes, such a deviation will not change her trading strategy. Thus, the total demand process $Y$ remain unchanged.  

Then the  terminal wealth of the deviating market maker is given by
\be \label{e:gainsj}
G^j_1(X):=-\int_0^{1} Y^{j}_{s-}\,dP_s +Y^{j}_{1} (P_{1}-V),
\ee
where $P$ is the price process corresponding to the market quote $(\lambda, \phi)$ and $Y^j$ is as in \eqref{e:Yjdef} with $Y=X+Z$. Thus, such a deviation is not profitable if $ \bbE\left(U\left(G^j_1(X)\right)\right)$ is less than the expected utility had the market maker $j$  quoted $(\lambda,\phi)$.

If $\tilde{\lambda}\neq \lambda$, i.e. the market quote changes, the insider's optimal strategy may change. Denote by  $\cX(\tilde{\lambda},\phi)$ the set of optimal strategies of the insider given this market quote  and fix $\tilde{X}\in \cX(\tilde{\lambda},\phi)$. Then, the  terminal wealth  $G^j_1$ of the deviating market maker is given by  \eqref{e:gainsj}, where $P$ is the price process corresponding to the market quote $(\tilde{\lambda}, \phi)$ and $Y^j$ is as in \eqref{e:Yjdef} with $Y=\tilde{X}+Z$. We will say that this deviation is not profitable if the expected utility of the deviating market maker from terminal wealth is lower for all $\tilde{X}\in \cX(\tilde{\lambda},\phi)$.

Note that in our definition of equilibrium, the market makers take into account the response of the insider when considering deviations. That is, the equilibrium that we consider is of multi-leader Stackelberg-Nash-Bertrand type in which the group of market makers is the leader and the insider is the follower. The market makers compete on price in a non-cooperative Nash game while taking into account the optimal response of the insider.

\subsection{Equilibrium} \label{s:eq}
The feedback from the insider's strategy makes the analysis of the actions of the deviating market maker quite challenging. Nevertheless, the nature of the optimal response of the insider as given by  Proposition \ref{p:ioptimality} implies  that  deviations in market quotes cannot result in a decrease in $\lambda$ on $[0,1)$ as this will bring an infinite loss to the deviating market maker. Note that the competition among the market makers at time $1$ becomes relevant only if there is a bulk order. This trade can only come from the insider whose gains from the bulk order $\delta$ is given by $\delta (V-P_1)$. As the insider choses $\delta$ to maximise her utility, at the optimal bulk order her gains is strictly positive. Hence, the market maker taking the other side of the trade has strictly negative gains since there is no noise trading at time $1$. This prevents the market makers with zero inventory making the market at time $1$. On the other hand, for the market makers with non-zero inventory, participating in this final trade could be optimal as it might reduce their inventory risk. 

The different nature of the market at time $1$ can result in a decrease of  the market $\lambda$ at closing of the market. However, there is a minimal value for $\lambda$ below which absorbing the final {\em bulk} order will result in infinite loss.

Therefore, as soon as the final $\lambda$ is set at this minimal value, the deviating market  maker  has  only two options: 1) quote $\lambda^*$ at the beginning and quit making the market at some point, and 2) undercut at time $0$ and possibly leave market at some future time.

Note that the market quote (and thus the insider's strategy) remains the same in the first case. Consequently, deviations of this form will be suboptimal provided that the value of leaving market is a martingale. 

Although the market quotes and the corresponding insider strategies change in the second case, any optimal response of the insider leads to the full revelation of the fundamental value at the time the deviating market maker stops making the market. As the market maker's utility increase in $\lambda$, the deviating market maker's utility is negative as soon as the equilibrium one is zero. 

These two observations lead to the following theorem, proof of which is delegated to the appendix.

\begin{theorem} \label{t:eqMM}
	There exist two equilibria $((\Lambda^*,\Phi^*_i,X^*_i))_{i=1,2}$, where
	\bean
	\Lambda^*(t)&=&\Big(\lambda^*\chf_{[t<1]}+ \frac{\lambda^*}{2}\chf_{[t=1]}\Big)\mathbf{1}, \mbox{ where } \lambda^*=\frac{\rho\gamma^2}{2}+\sqrt{\frac{\rho^2\gamma^4}{4}+\frac{\gamma^2}{\sigma^2}},\\
	\Phi^*_i&=&\phi^*_i\mathbf{1}, \mbox{ where }\phi^*_i=\mu-\frac{(-1)^i}{\rho}\sqrt{\rho^2\gamma^2+2\frac{\rho\gamma^2}{\lambda^* \sigma^2}\log\left(\frac{\gamma}{\lambda^* \sigma}\right)},\\
	dX^*_{i,t}&=&\sigma^2 \left(a(t) (X^*_{i,t}+Z_t)+b_i(t) +\frac{p_x^{(i)}}{p^{(i)}}\Big(t,X^*_{i,t}+Z_t;1, \frac{V-\phi^*_i}{\lambda^*}\Big) \right)dt, \mbox{ where }\\
	a(t)&=&-\frac{\lambda^*\rho}{1+ \lambda^* \rho\sigma^2(1-t)},\\
	b_i(t)&=&(-1)^i\frac{1}{2}\log(1+\lambda^* \rho\sigma^2(1-t))\sqrt{\frac{a(t)}{-\lambda^* \rho\sigma^2 (1-t)+\log (1+\lambda^* \rho\sigma^2(1-t))}},
	\eean
	and $p^{(i)}$ is the transition density of the Ornstein-Uhlenbeck process
	\be \label{e:Ri}
	dR^{(i)}_t=\sigma dB_t +\sigma^2 (a(t)R^{(i)}_t+b_i(t))dt.
	\ee
	\begin{enumerate}
		\item From the point of view of the insider the equilibrium demand and price dynamics are given by
		\bea
		dY^{*,i}_t&=&\sigma dB_t+\frac{1}{1-t}\left(\frac{V-\phi_i^*}{\lambda^*}-\frac{\alpha_{1,i}(t)}{\rho\lambda^*}+ b_i(t)\sigma^2(1-t)-Y^{*,i}_t\right)dt,\label{e:demandMMi}\\
		dP^{*,i}_t&=&\lambda^*\sigma dB_t + \frac{1}{1-t}\left(V-\frac{\alpha_{1,i}(t)}{\rho}+ b_i(t)\lambda^*\sigma^2(1-t)-P^{*,i}_t\right)dt\label{e:priceMMi},
		\eea
		where
		\[
		\alpha_{1,i}(t)=(-1)^i\sqrt{-a(t)\left(\lambda^* \rho\sigma^2 (1-t)-\log (1+\lambda^* \rho\sigma^2(1-t))\right)}.
		\]
		The ex-ante profit of the insider is given by
		\[
		\bbE\left[\frac{1}{2 \lambda^*}(P_0-V)^2+ \frac{\lambda^*\sigma^2}{2}\right]=\frac{(\phi_i^*-\mu)^2}{2\lambda^*}+ \frac{\gamma^2}{2\lambda^*}+ \frac{\lambda^*\sigma^2}{2}.
		\]
		\item From the point of view of the market makers the equilibrium demand and price dynamics are given by
		\bea
		dY^{*,i}_t&=&\sigma d\beta_t +\sigma^2 (a(t)Y^{*,i}_t+b_i(t))dt,\label{e:demandMMm}\\
		dP^{*,i}_t&=&\lambda^*\sigma d\beta_t +\sigma^2 (a(t)(P^{*,i}_t-\phi_i^*)+\lambda^*b_i(t))dt\label{e:priceMMm},
		\eea
		where
	$\beta$ is an $\cF^M$-Brownian motion. 
	\end{enumerate}
\end{theorem}
Comparison of  this result with Theorem \ref{t:eqLP} reveals that the equilibrium market parameters are the same apart from $P_0$ and the mean reversion level of total demand. As the target level of inventory for competitive agents is $0$, we conclude that competition among makers results in them taking additional inventory risk. The inventory risk depends on the magnitude of inventory and not its sign giving raise to two equilibria, with the magnitude of the target inventory determined by the zero-utility condition. This phenomenon is not an artefact of asymmetric information as it appears in the equilibrium with uninformed strategic stated below. As in the case of competitive agents, the speed of mean reversion is smaller in absence of private information due to smaller perceived risk by the liquidity suppliers.
\begin{theorem} \label{t:benchMM}
There exist equilibria $(\Lambda^*,\Phi^*_i,X^*_i)_{i=1,2}$, where
\bean
\Lambda^*&=&\lambda^* \mathbf{1}, \mbox{ with } \lambda^*(t)=\frac{\rho\gamma^2}{2}\chf_{[t<1]}+\frac{\rho\gamma^2}{4}\chf_{[t=1]},\\
\Phi^*_i&=&\phi^*_i\mathbf{1}, \mbox{ with }\phi^*=(-1)^i\frac{\rho \sigma \gamma^2}{2\sqrt{3}}+\mu,\\
dX^*_{i,t}&=&\chf_{[t<1]}\sigma^2 \left(a(t) (X^*_{i,t}+Z_t)+b_i(t)\right)dt +\chf_{[t=1]}\frac{\mu-\phi^*_i-\lambda^*(1-)(X_{i,1-}^*+Z_1)}{\lambda^*(1-)}, \mbox{ where }\\
a(t)&=&-\frac{\rho^2\gamma^2}{4},\\
b_i(t)&=&-\frac{\rho(\mu-\phi^*_i)}{2}.
\eean
Moreover, the equilibrium dynamics of the demand is given by
\[
\begin{split}
	dY^*_{i,t}&=\sigma dB_t -\frac{\rho\sigma^2}{2}\Big(\frac{\rho\gamma^2}{2}Y^*_{i,t}+\mu-\phi^*_i \Big)dt, \quad t<1;\\
	\Delta Y^*_{i,1}&=2\frac{\mu- \phi^*_i}{\rho \gamma^2}-Y_{i,1-}^*,
\end{split}
\]
and the expected profit of the strategic trader in equilibrium equals 
$
\frac{\rho \sigma^2 \gamma^2}{3}$.
\end{theorem}

\section{Comparative dynamics} \label{s:cd}
{ In this section we consider the case $V\sim N(0,\gamma)$, i.e. $\mu=0$. Our aim is to   analyse the impact of risk aversion on the market liquidity, in particular {\em depth} and {\em resilience}, {\em price efficiency}, {\em reversal} of the equilibrium prices, {\em insider's profit}, and the value of private information, i.e. the difference between the profits of the insider and a strategic trader with no information. Recall that depth corresponds to the order size necessary to move the price by one unit; the resilience is the speed with which the prices converge to the fundamental value; price efficiency is a measure of remaining uncertainty in the prices; and, finally, the reversal reflects the autocorrelation of returns.

Our comparative analysis will show that there is single quantity that parametrises the above measures of liquidity across markets that differ on the risk aversion of liquidity suppliers, inventory and adverse selection risks.  More precisely, all measures of liquidity when normalised by their Kyle (risk-neutral) counterparts will only depend on {\em the market adjusted risk aversion}  parameter $\rho_M:=\rho\gamma\sigma$.  One way to understand the universality of this parameter is to observe that the utility of liquidity providers from wealth $W$ is $1-\exp(-\rho_M\frac{W}{\gamma\sigma})$. Note that $\frac{W}{\gamma\sigma}$ is dimensionless since the unit of $W$ is $\mbox{time}\times \mbox{currency}$ whereas $\gamma^2$ is in units of $\mbox{time}\times \mbox{currency}^2$ and $\sigma^2$ is in units of $\mbox{time}$.  By making the risk dimensionless we can shift our attention from idiosyncratic features of each market towards the general mechanism underpinning  the market liquidity.  Although in what follows we shall consider all values of $\rho_M$, the interval $[0,1)$ is of particular interest since it encompasses all reasonable values observed in markets\footnote{As dividends are less volatile than asset prices, if we consider risk aversion parameter less than 10 and yearly volatility less than 0.3, we will be in that range.}.

For the equilibria considered in Sections \ref{s:LP} and \ref{s:MM} we will have two benchmark models to compare to. The first model is the one with risk-neutral market makers  as in \cite{Kyle}, in which case the depth of the market is constant and equals $\frac{1}{\lambda_K}:=\frac{\sigma}{\gamma}$, resilience is a function of time and is given by $R(t)=\frac{1}{1-t}$, efficiency equals $\Sigma_K(t):=\gamma^2(1-t)$, and  the ex-ante profit of the insider  is  $\pi_K:=\gamma \sigma$. We will also consider\footnote{The additional normalisation by $\eps$ is an artefact of continuous time models. For $t$ and $u$ close to $s$, the covariance is of order $(t-s)^2$, thus, a further normalisation is needed since the variances are also of order $t-s$.} \be \label{e:momentum}
M(s):=\lim_{\substack{t-s=\eps \\u-s=\eps\\ \eps \rar 0}}\frac{\mbox{Cov}_s\left(P_t-P_s, P_u-P_t\right)}{\sqrt{\mbox{Var}_s(P_t-P_s)\mbox{Var}_s(P_u-P_t)}\eps},
  \ee
where $\mbox{Cov}_s$ and $\mbox{Var}_s$ are covariance and variance given market's information at time $s$. If $M$ is positive, it indicates a price momentum whereas a negative value implies price reversal. In Kyle's model $M$  is equal to $0$ at all times.  

We will also compare our models to the benchmark models with strategic trader  described in the previous section. As the strategic trader does not possess any private information, the equilibrium prices do not carry any information about the fundamental value, and therefore the resilience is undefined and efficiency is constant, $\gamma^2$. However, we can calculate the market depth and the profit of the strategic trader.  According to Theorems  \ref{t:benchLP} and \ref{t:benchMM}  the market depth is $\frac{1}{\lambda_0}:=\frac{2}{\rho \gamma^2}$ in both  models.  While the  profit of the strategic trader equals $\frac{\rho\gamma^2\sigma^2}{4}$ in the market with competitive agents, her profit becomes $\frac{\rho \gamma^2 \sigma^2}{3}$ when she trades against market makers.  Thus, it is evident that given a choice the strategic trader would choose  market  makers to execute their orders. 

To analyse the relative impact of inventory and private information on price formation, we need to measure the change in price impact  as a result of private information.  Since the equilibrium market depth does not depend on the market type,  we shall use the results  of Theorems \ref{t:eqMM} and  \ref{t:benchMM} to conduct our analysis.  We denote by $\lambda$ and $\lambda_0$ the  inverse market depths in case of insider and  strategic (and uninformed) trader correspondingly. It is evident that $\lambda>\lambda_0$. Thus, the market depth decreases when the strategic trader gains access to private information. This decrease is  due to asymmetric information and we will use it as a measure of the component of the  equilibrium price affected by  adverse selection. Direct calculations show that the adverse selection component, $\lambda -\lambda_0=\sqrt{\frac{\rho^2\gamma^4}{4}+\frac{\gamma^2}{\sigma^2}}$ is  increasing in risk aversion and decreasing in noise volatility. This is intuitive since the liquidity providers require higher compensation for facing adverse selection as their risk aversion increases. Moreover, higher noise volatility makes total demand less informative to the liquidity providers which decreases the informational component of the price adjustment as well as its proportion. Naturally, one expects that the cost of holding an inventory increases in risk aversion, even in the absence of private information, thus the component of the price determined by purely inventory considerations is increasing in risk aversion, too.   In Figure \ref{fig:advcomp} we plot  $(\lambda -\lambda_0)/\lambda$ (proportion of price due to adverse selection) as a function of market adjusted risk aversion. As one can observe,  the increase in the cost of inventory holding is strong enough to reduce the proportion of price affected by  private information. However, this inventory component never dominates, and accounts for a half of the price differential at most.

\begin{figure}[h]
\centering
\includegraphics[height=6cm, width=11cm]{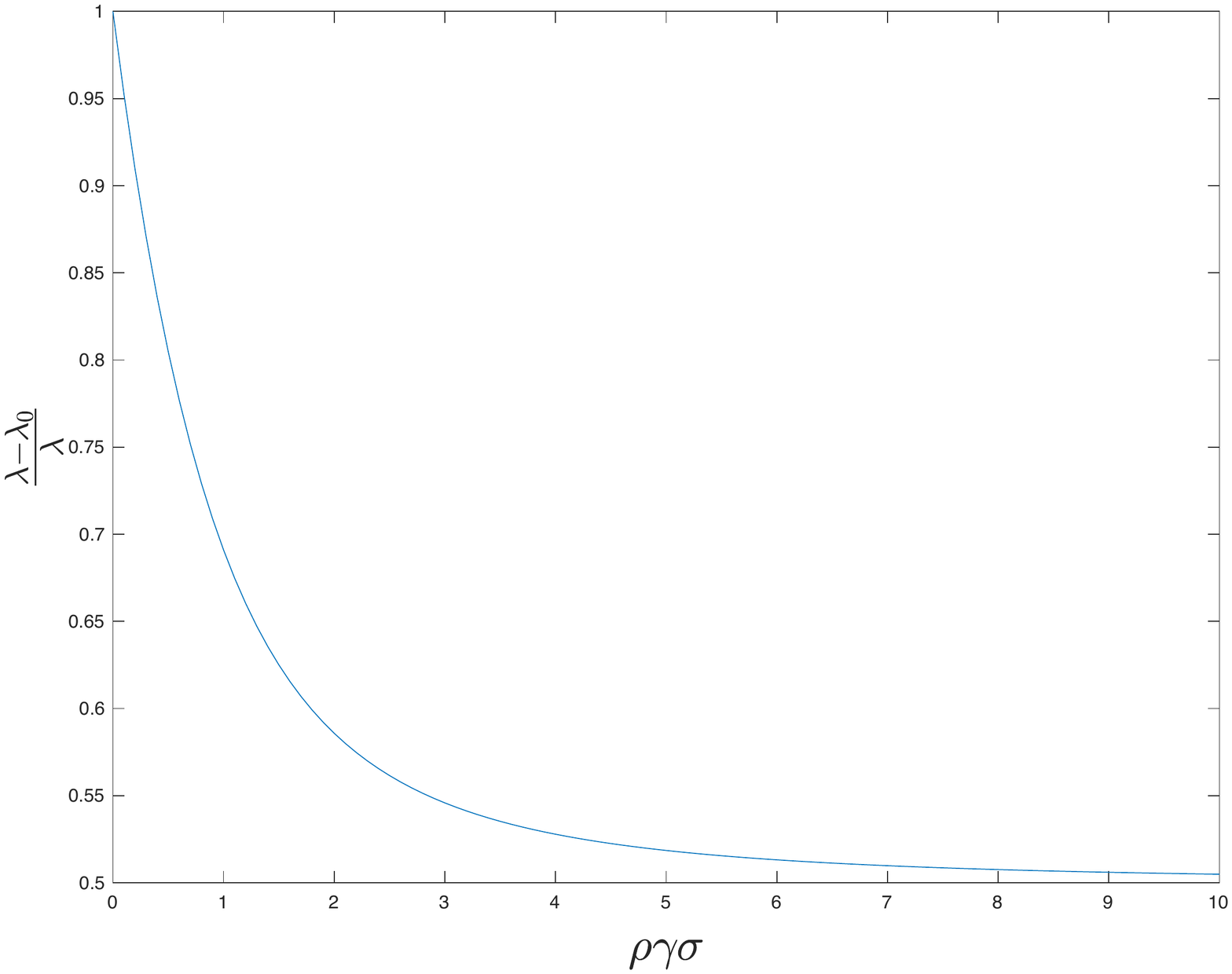}
\caption{The figure shows the dependency of the adverse selection proportion, which is defined to be  $(\lambda- \lambda_0)/\lambda$, of the equilibrium price on the market adjusted risk aversion. }
\label{fig:advcomp}
\end{figure}

\begin{table}[tb]
	\caption{Market depth}
	\label{T:depth}\vspace{1mm}
	\par
	\begin{center}
		
		\bgroup
		\def\arraystretch{1.25}
		
		\begin{tabular}{ccc}
			& Competitive Agents & Market Makers  \\
			\hline
			Insider trader & $\frac{1}{\lambda_K}\frac{1}{\frac{\rho_M}{2}+\sqrt{\frac{\rho_M^2}{4}+1}}$ & $\frac{1}{\lambda_K}\frac{1}{\frac{\rho_M}{2}+\sqrt{\frac{\rho_M^2}{4}+1}}$\\
			\hline
			Strategic trader & $\frac{2}{\rho \gamma^2}=\frac{1}{\lambda_K}\frac{2}{\rho_M}$	&  $\frac{2}{\rho \gamma^2}=\frac{1}{\lambda_K}\frac{2}{\rho_M}$   \\
			\hline
		\end{tabular}
		
		\egroup
		
	\end{center}
\end{table}

As observed above, in  both equilibria with an insider  the market depth  is the same (for convenience we report the market depth in all four equilibria in Table \ref{T:depth}) and smaller than $\frac{1}{\lambda_K}$. This  implies that the depth of the market with risk averse liquidity suppliers is smaller than the depth observed in \cite{Kyle}.   The reason for the observed decrease in the market depth is the risk sharing between the liquidity suppliers and the insider as the risk aversion makes the liquidity providers require additional compensation for large quantities of risk. This effect gets more pronounced as  the risk aversion coefficeint $\rho$ increases and  the market approaches to complete illiquidity when the risk aversion rate gets high.   One can also show that the market depth increases with the volatility of noise trading since the total order process contains less information and therefore the price is less sensitive to it. 

On the other hand, if we compare the loss in liquidity due to risk aversion, that is, the depth with risk averse liquidity providers normalised by the one in \cite{Kyle},  we observe  in Figure \ref{fig:depth1} that the loss in liquidity become larger as the market adjusted risk aversion parameter increases, i.e. the market conditions get riskier from the point of view of liquidity providers.  The Figure \ref{fig:depth1}  emphasises that to study  the market liquidity behaviour we must focus on the collective impact of market parameters via $\rho_M$ rather than their effect in isolation.  For instance, by keeping the risk aversion $\rho$ and $\frac{\gamma}{\sigma}$ constant, we can obtain strikingly different market behaviour from complete illiquidity when $\gamma$ and $\sigma$ are large to liquidity levels of risk-neutral liquidity suppliers for small $\gamma$ and $\sigma$. 

\begin{figure}[h]
\centering
\includegraphics[height=6cm, width=10cm]{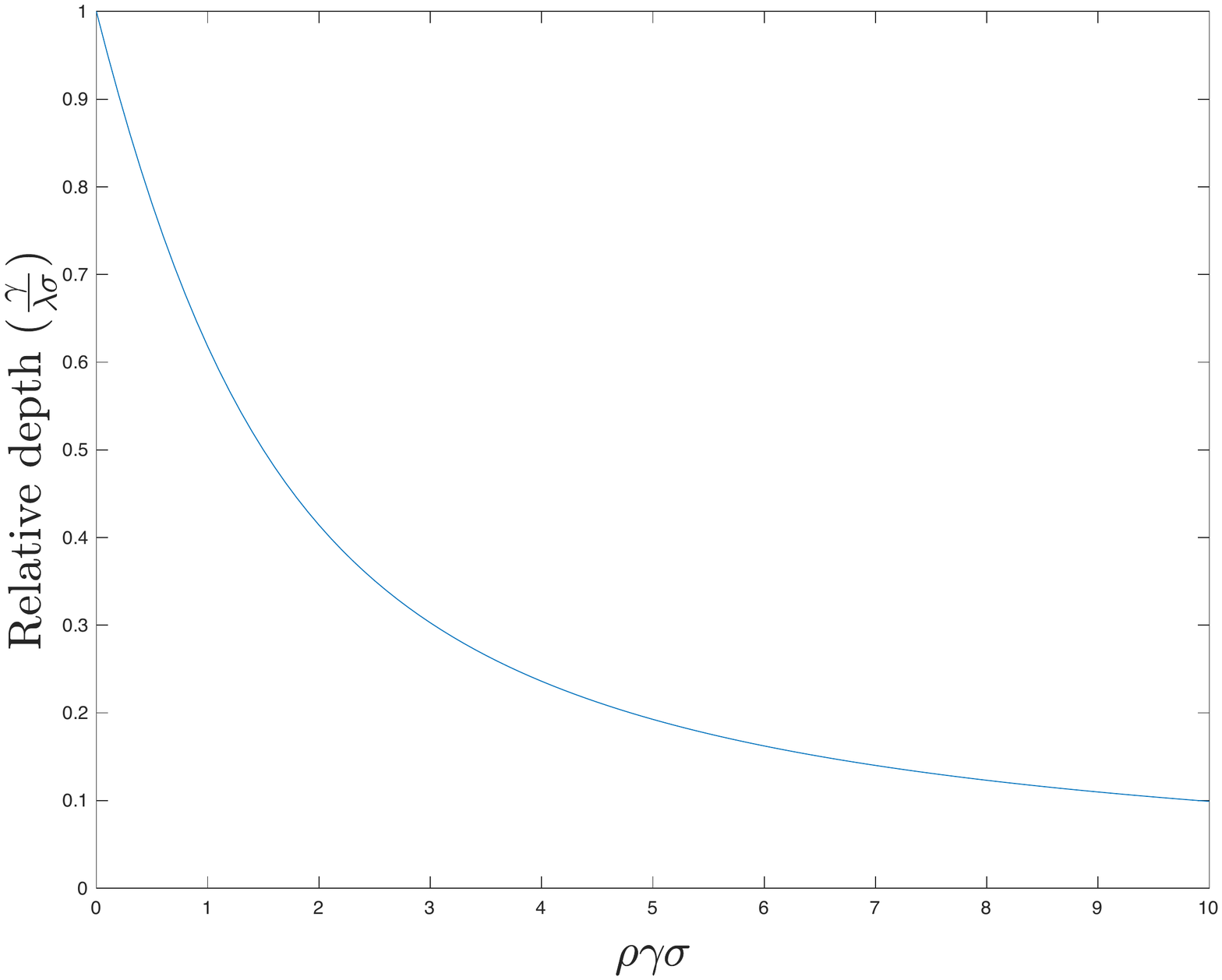}
\caption{The figure depicts how the relative depth, i.e. $\frac{\gamma}{\lambda \sigma}$ depends on $\rho_M$.}
\label{fig:depth1}
\end{figure}

It is evident from price dynamics (\ref{e:priceLP}) and (\ref{e:priceMMi}) that in both equilibria convergence to the fundamental value occurs at rate $\frac{1}{1-t}$, which is the same level of resilience as  \cite{Kyle}. Thus, we can conclude that the resilience is unaffected by risk aversion or market design. 

Efficiency is a measurement of how informative the market prices are. Using  Lemma \ref{l:martingaleU}, we get
\begin{align*}
\Sigma(t)&=\mbox{Var}(P_1|\cF^M_t)=\mbox{Var}(\lambda Y_1|Y_t)=\frac{\lambda^2\sigma^2(1-t)}{1+\lambda\rho\sigma^2(1-t)}\\
&=\Sigma_K(t)\frac{2+\rho_M(1+\sqrt{\rho_M^2+4})}{2+\rho_M(1+\sqrt{\rho_M^2+4})(1-t)}
\end{align*}
in both market makers and competitive agents equilibrium. 
 
Direct differentiation show that $\Sigma$ is decreasing in $t$, which reflects the fact that prices become more informative as time progresses as a result of the insider's effort to drive prices to the fundamental value. Moreover, it is clear   that $\Sigma(t) \geq \Sigma_K(t)$, i.e. the prices become less efficient  when the  liquidity providers become risk averse.  This behaviour is not completely driven by $\rho$: as Figure \ref{fig:efficiency} illustrates  the loss of efficiency is monotone in $\rho_M$. In particular, as $\rho_M \rar \infty$, $\Sigma$ behaves like $\gamma^2$, independent of $t$, indicating a complete loss of efficiency.

The monotonicity in $\rho_M$ is to be expected as the insider has to participate in risk sharing to make the equilibrium possible as the market gets riskier for the liquidity providers. She does so by forcing the total demand to mean revert around a time varying level which provides  sufficient level of risk sharing. More precisely, the total demand dynamics in both equilibria is given by
\be \label{e:demandparm}
dY_t= \sigma d\beta_t + \sigma^2 a(t)\left(Y_t +\frac{b(t)}{a(t)}\right)dt,
\ee
where 
\[
\sigma^2 a(s)=-\frac{2}{1 +\sqrt{1 +\frac{4}{\rho_M^2}}-2s},
\]
 $b=0$ for the competitive agents equilibrium or $b_i(t)$ in the market makers one.  

 It is evident that the level of mean reversion is determined by the liquidity suppliers, and thus carry no information about the fundamental value.   This results in a slower convergence to the fundamental value, which in turn implies lower efficiency. When the market adjusted risk aversion parameter increases, liquidity suppliers demand more risk sharing from the insider and, therefore, stronger mean reversion, which manifests itself in $|a|$ being increasing in $\rho_M$. As a result, the total order stays closer to the long run mean, which makes the speed of convergence to the fundamental value  decrease with $\rho_M$, i.e.  the market becomes less efficient when market becomes riskier for the  liquidity suppliers.  
 
This loss of efficiency and the associated off-load of the risk by the liquidity providers becomes particularly striking  as $\rho_M$ increases to infinity. In this case the scaled equilibrium demand, $\tilde{Y}:=\frac{Y}{\sigma}$, converges to the solution of 
  \[
  d\tilde{Y}_t = dB_t -\frac{\tilde{Y}_t}{1-t}dt
  \]
 since $\sigma b_i$ vanishes in the limit. As the solution to the above equation is a standard Brownian bridge with $\tilde{Y}_1=0$, this implies that, after we control for the noise volatility, {\it all} inventory risk is shifted from the liquidity providers to the insider and {\it all} asymmetric information risk is shifted to the noise traders.
 
This off-load of liquidity risk by liquidity suppliers also features in both strategic trader equilibria. Indeed, in the model with market makers simple calculations reveal that
\[
d\tilde{Y}_t=dB_t-\frac{\rho_M^2}{4}(\tilde{Y}_t \pm \frac{1}{\sqrt{3}})dt,
\]
the solution of which is $\tilde{Y}_1= \pm \frac{1}{\sqrt{3}} +e^{-\frac{\rho_M^2}{4}}\left(\pm \frac{1}{\sqrt{3}}+\int_0^1e^{\frac{\rho_M^2s}{4}}dB_s\right)$. As the market adjusted risk aversion parameter gets larger $\tilde{Y}_1\rar \pm \frac{1}{\sqrt{3}} $ -- a level of inventory that provides zero expected utility. 

On the other hand, if the liquidity is supplied by competitive agents,  the limiting dynamics of $\tilde{Y}$  as $\rho_M \rar \infty$ is given by
\ \[
d\tilde{Y}_t = dB_t -\frac{\tilde{Y}_t}{1-t}dt
\]
when the strategic trader is uninformed. Thus, once again, we get $\tilde{Y}_1=0$ and the competitive agents carry no liquidity risk in the limit.

As higher volatility of the noise trades implies that liquidity suppliers face higher inventory risk and therefore require more risk sharing which, in turn, slows down price discovery, one would expect that the market efficiency is affected by $\sigma$. The above discussion indeed supports this observation as  $\rho_M$ is proportional to $\sigma$.  This is in stark contrast with the risk neutral case, where the efficiency is independent of the volatility of the noise trading as the inventory risk is not priced in.

\begin{figure}[h]
	\centering
	\includegraphics[height=6cm, width=10cm]{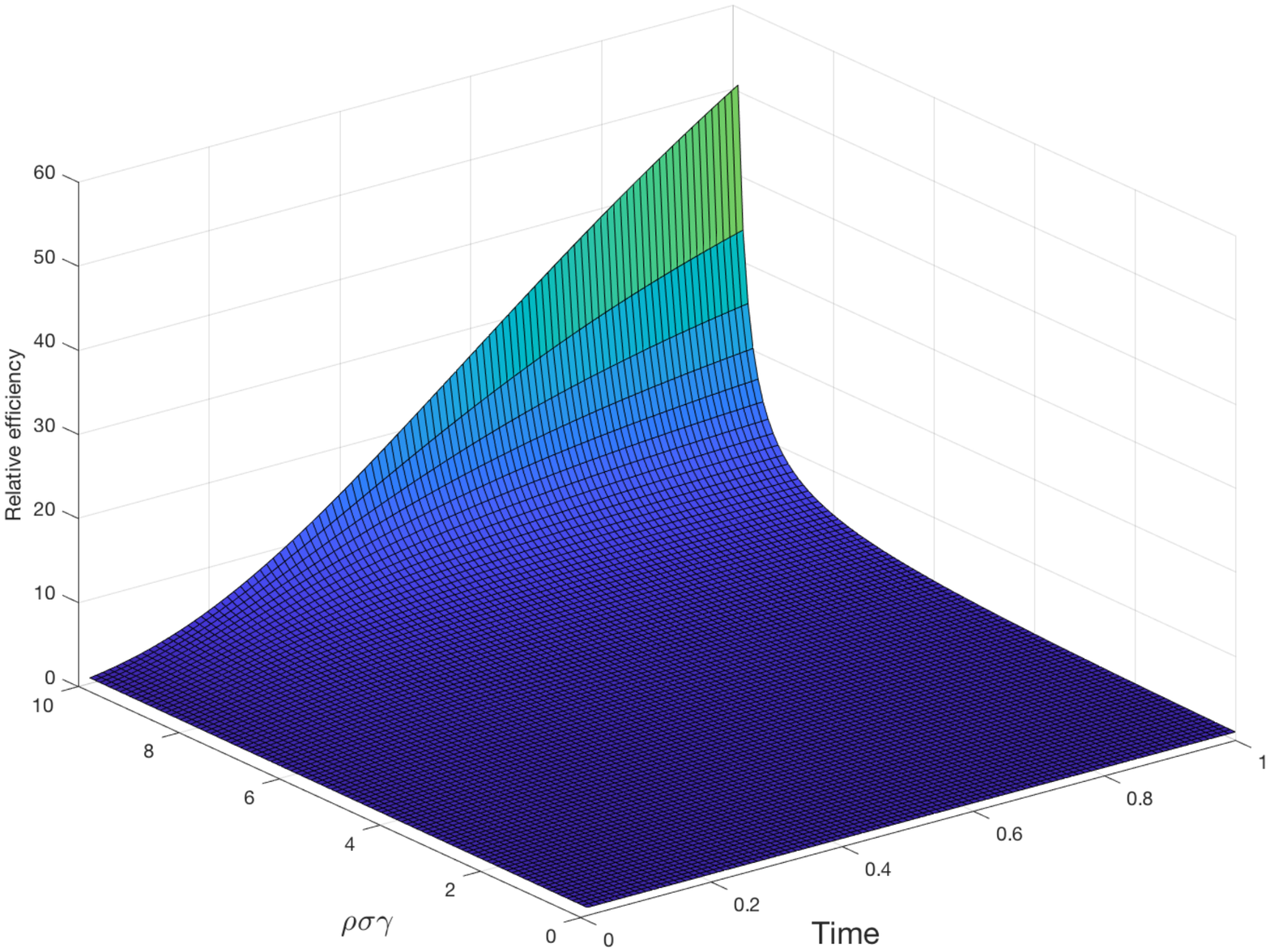}
\caption{Relative efficiency is a measurement for residual uncertainty.  The above shows it as a function of time and market adjusted risk aversion $\rho_M$. }
	\label{fig:efficiency}
\end{figure}

When the liquidity suppliers become risk averse, the prices exhibit  reversal.  That is, the large buy orders are usually followed by sell orders and vice versa. This is apparent from the evolution of equilibrium demand given in (\ref{e:demandparm}). To compute the intensity of reversal, observe that in all equilibria $P$ is of the form $P=\lambda Y +\phi$ and therefore
\bean
M(s)&:=&\lim_{\substack{t-s=\eps \\u-s=\eps\\ \eps \rar 0}}\frac{\mbox{Cov}_s\left(Y_t-Y_s, Y_u-Y_t\right)}{\sqrt{\mbox{Var}_s(Y_t-Y_s)\mbox{Var}_s(Y_u-Y_t)}\eps}\\
&=&\lim_{\substack{t-s=\eps \\u-s=\eps\\ \eps \rar 0}}\frac{\bbE_s\left((Z_t-Z_s)(Z_u-Z_t)\right)}{\sqrt{\mbox{Var}_s(Z_t-Z_s)\mbox{Var}_s(Z_u-Z_t)}\eps},
\eean
where $Z_t= Y_t-\bbE_s[Y_t]$. Direct calculations yield
\[
M(s)=\sigma^2 a(s)=-\frac{1}{\frac{1 +\sqrt{1 +\frac{4}{\rho_M^2}}}{2}-s}
\]
in both markets when an insider is present.  Similar calculations show that $M(s)=-\frac{\rho_M^2}{4}$, thus, the reversal is constant in time   in the model with market makers and uninformed strategic trader. Finally, when the uninformed strategic trader trades against the competitive agent, we have
\[
M(s)=-\frac{\rho_M^2}{4+\rho_M^2(1-s)},
\]
i.e. prices exhibit reversal  in all models that we consider. 
For the convenience of the reader the intensity of price reversal, i.e.  $-M(s)$, as a function of $\rho_M$ and time for all four models is summarised in Table \ref{T:momentum} and is plotted in Figures \ref{fig:momentumLP} and  \ref{fig:momentumMM}.

\begin{table}[tb]
	\caption{Intensity of price reversal}
	\label{T:momentum}\vspace{1mm}
	\par
	\begin{center}
		
		\bgroup
		\def\arraystretch{1.25}
		
		\begin{tabular}{ccc}
			& Competitive Agents & Market Makers  \\
			\hline
			Insider trader & $\frac{2}{1 +\sqrt{1 +\frac{4}{\rho_M^2}}-2s}$& $\frac{2}{1 +\sqrt{1 +\frac{4}{\rho_M^2}}-2s}$\\
			\hline
			Strategic trader & $\frac{\rho_M^2}{4+\rho_M^2(1-s)}$	&  $\frac{\rho_M^2}{4}$   \\
			\hline
		\end{tabular}
		
		\egroup
		
	\end{center}
\end{table}

It is evident from Figures \ref{fig:momentumLP} and  \ref{fig:momentumMM} that the intensity of price reversal is an increasing function of market adjusted risk aversion parameter $\rho_M$.  The above discussion of market efficiency implies that this monotonicity is due to the risk sharing. Indeed, as observed above, risk sharing is responsible for stronger mean reversion in the demand as perceived risk faced by the liquidity providers grows, which in turn  implies that reversal should be stronger as market adjusted risk aversion increases. 

 Moreover, the intensity  of price reversal is non-decreasing in time.  This seems to be in conflict with  the risk-sharing idea between the insider and the liquidity suppliers since the remaining risk at time $0$ is higher.  However, taking into account that the liquidity suppliers derive their utility only from their consumption at time $1$, it becomes apparent that they place more emphasis on their inventory being close to the target at time $1$. Thus, risk-sharing  implies  stronger reversal at time $1$ as we observe from all models.  
 
 Although the price reversal seems to be affected by the nature of liquidity provision when the strategic trader is uninformed, if we focus on the dynamics close to market termination, we see that it does not depend on the type of the provision.  For  both types of liquidity suppliers the reversal  is stronger in the presence of insider due to the additional adverse selection risk and the resulting higher demand for risk sharing.
\begin{figure}[h]
	\begin{center}
		$
		\ba{cc}
		\scalebox{0.26}{\includegraphics{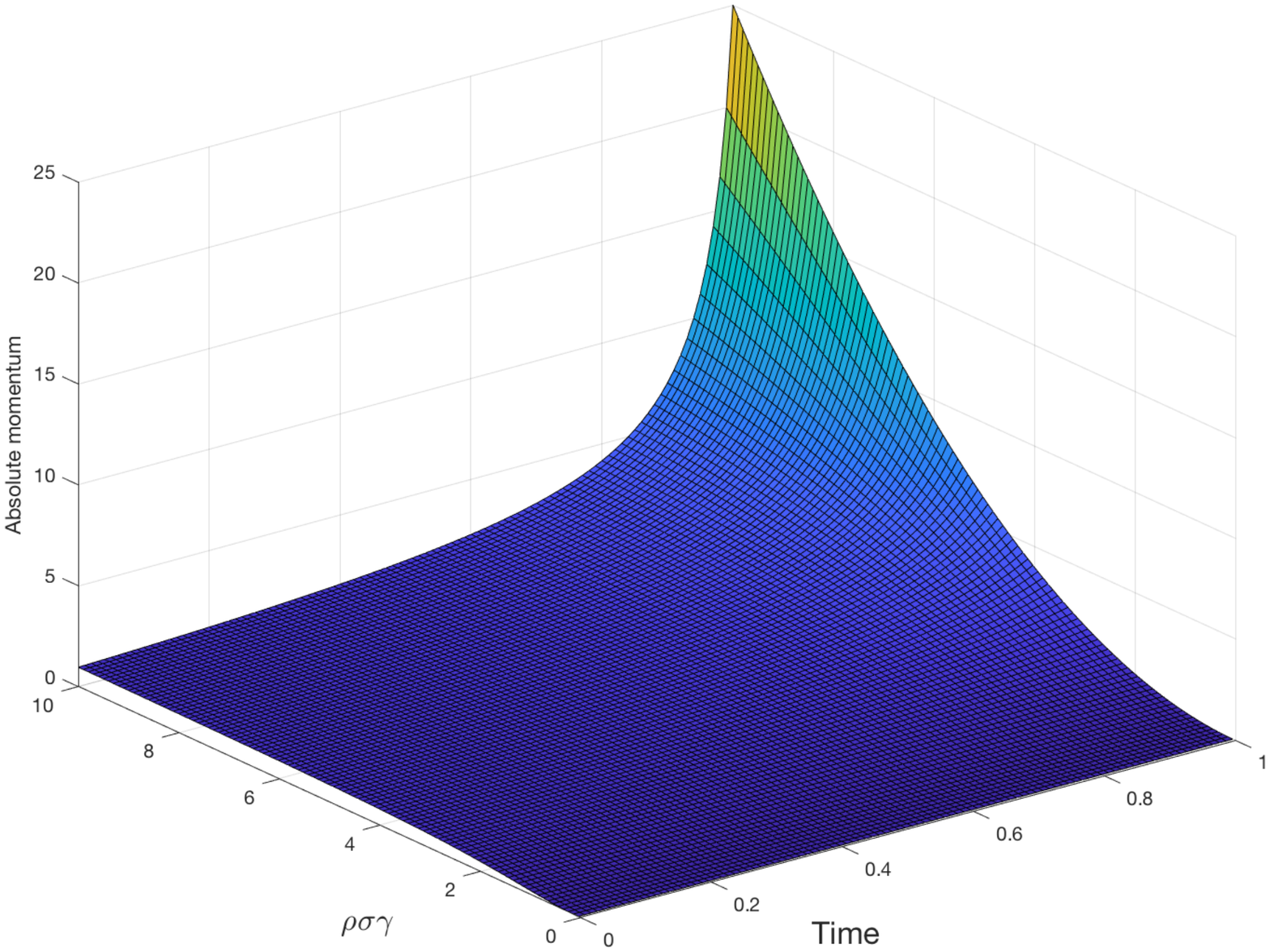}} & \scalebox{0.26}{\includegraphics{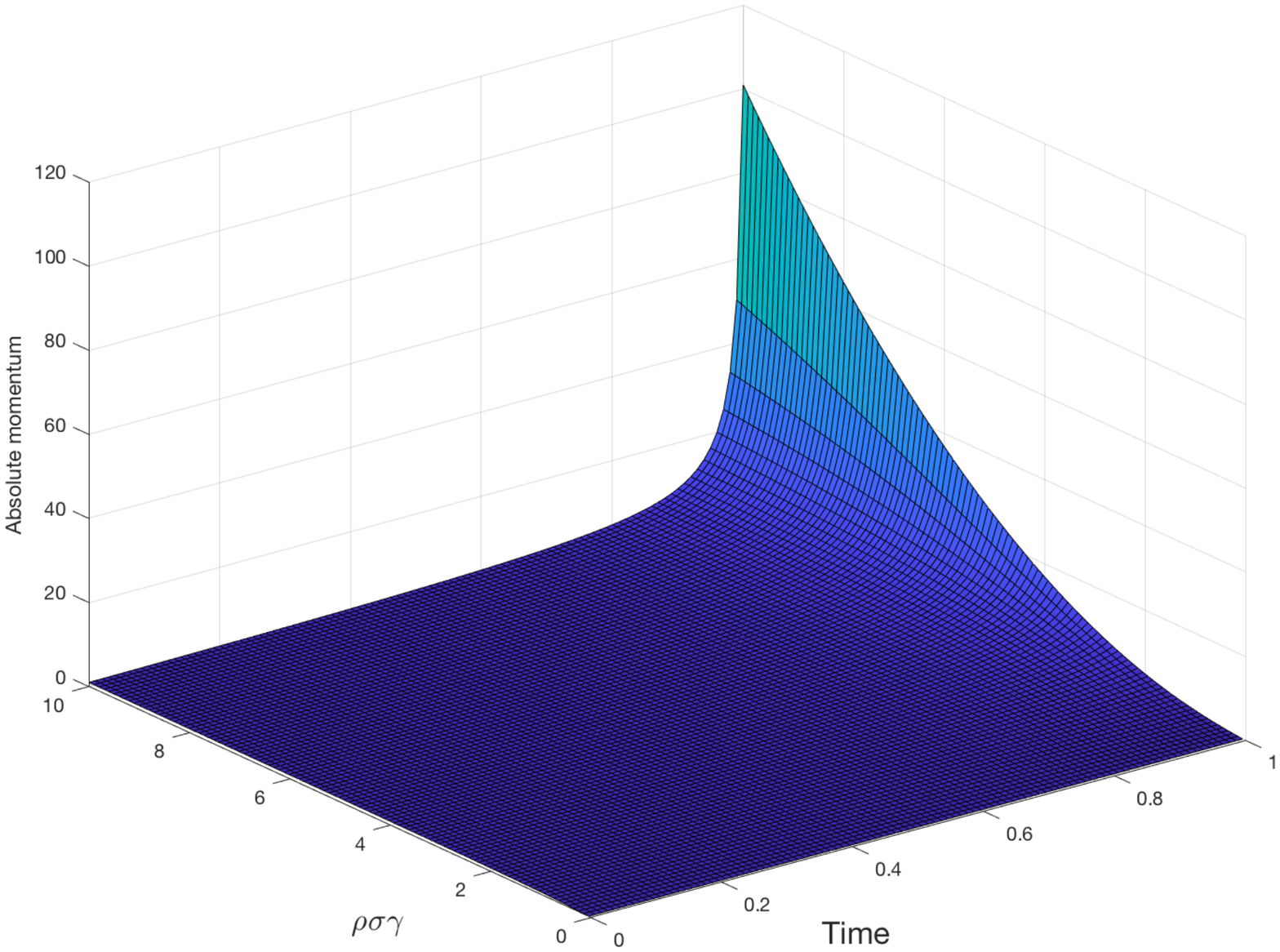}}
		\ea
		$
	\end{center}
	\caption{Intensity of price reversal in the market with competitive agents as a function of time and $\rho_M$. The left figure depicts price reversal in the presence of  strategic trader  and the right one in the   presence of insider. }
	\label{fig:momentumLP}
\end{figure}
\begin{figure}[h]
\begin{center}
$
\ba{cc}
\scalebox{0.26}{\includegraphics{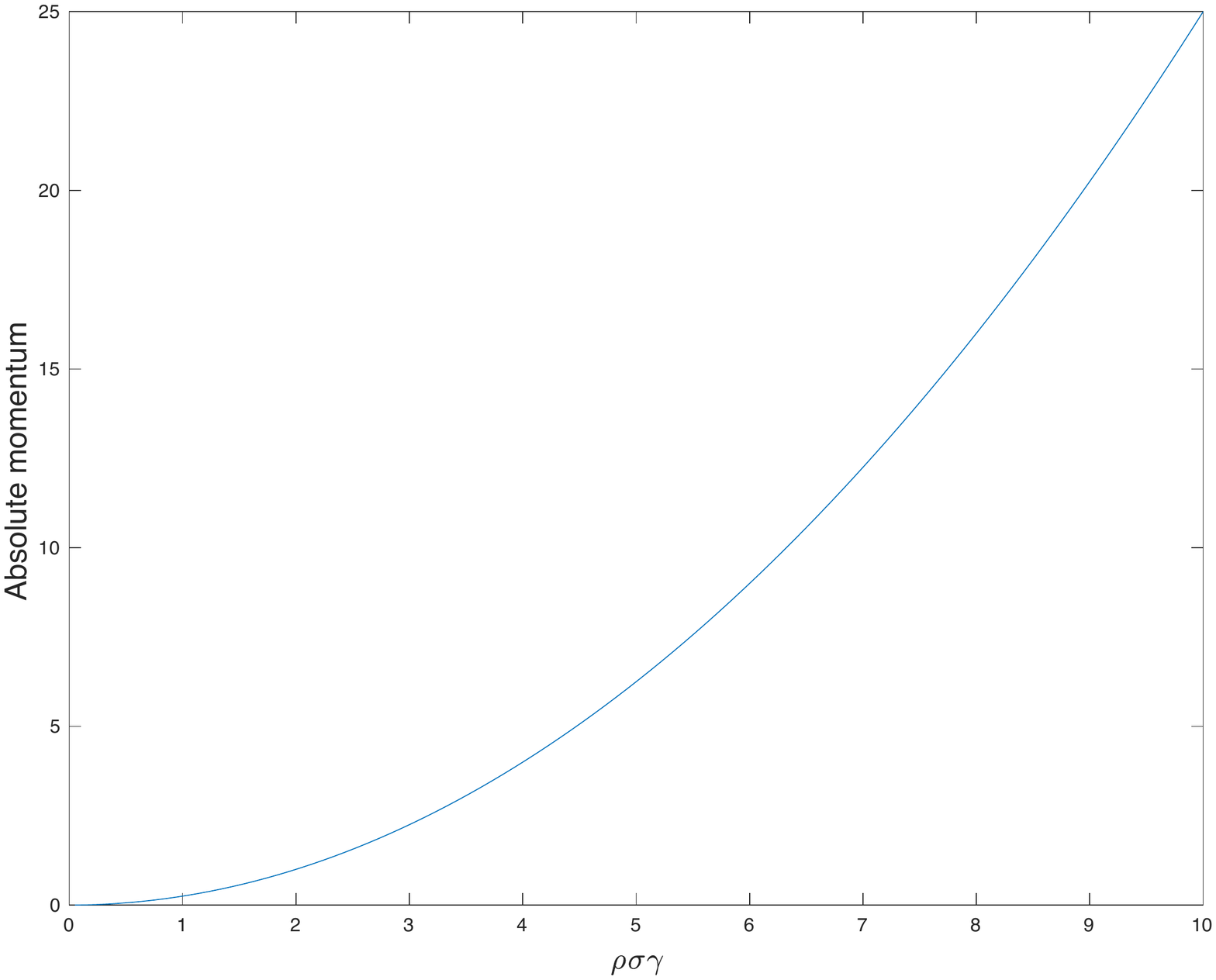}} & \scalebox{0.26}{\includegraphics{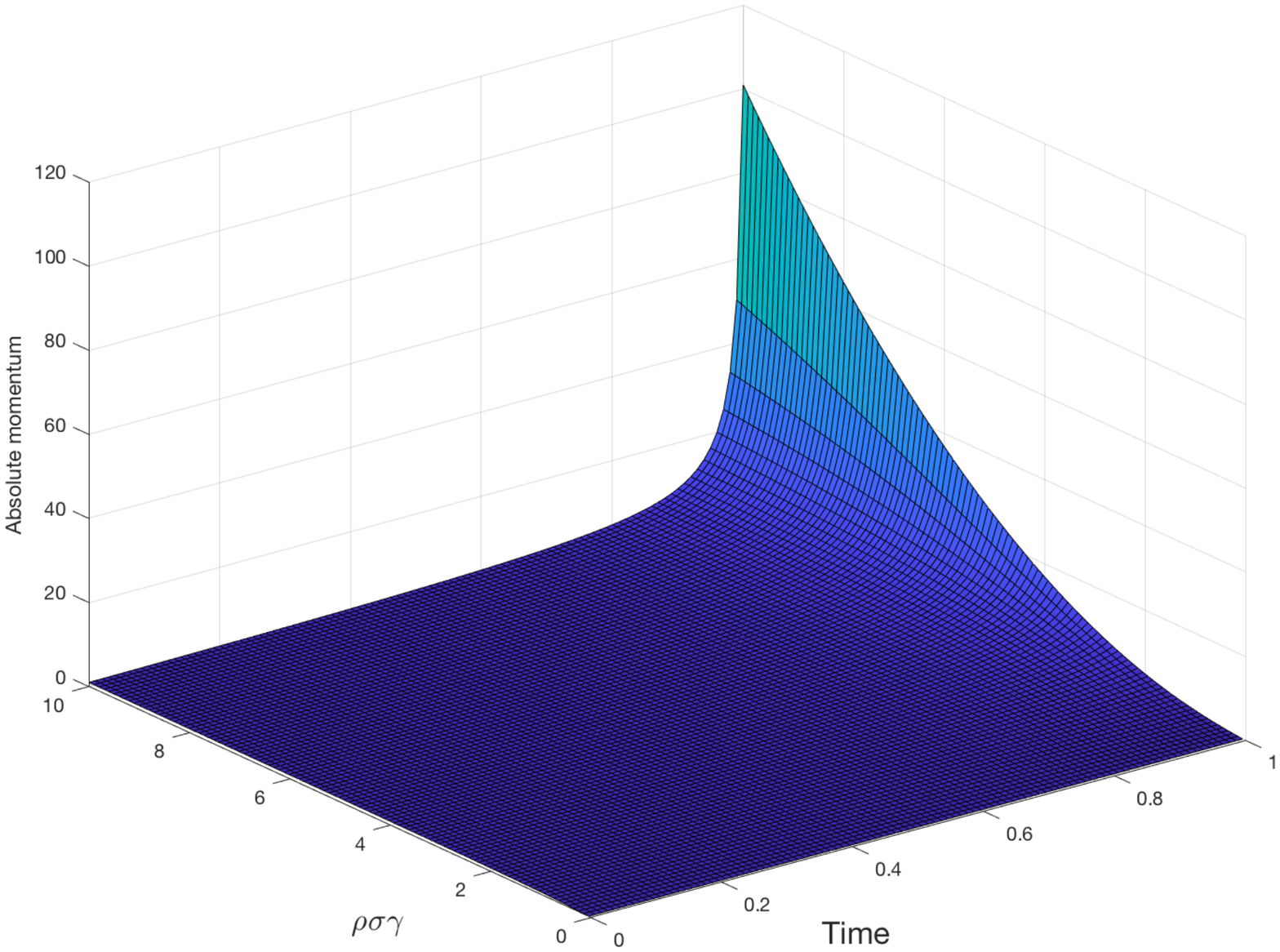}}
\ea
$
\end{center}
\caption{Intensity of price reversal in the market with market makers as a function of time and $\rho_M$. The left figure depicts price reversal in the presence of  strategic trader  and the left one in the   presence of insider. }
\label{fig:momentumMM}
\end{figure}

The last market characteristic we consider is the ex-ante profit of the insider. The expected profits of the insider, conditional on the value of $V$,  and the ex-ante profits of uninformed strategic trader  are summarised in Table \ref{T:profit}. This table also reports  the value of information in both models of liquidity provision normalized by the corresponding value in the Kyle model. 

 Observe that in both markets the profit of the insider can be expressed as 
\be \label{e:xpostIp}
\frac{1}{2\lambda }\left(P_0-V\right)^2+\frac{\sigma^2}{2}\lambda=\frac{1}{2\lambda }\left(\Delta^{P_0}-\Delta^V\right)^2+\frac{\sigma^2}{2}\lambda,
\ee
where $\Delta^{P_0}=P_0-\mu$ and $\Delta^V=V-\mu$. Since $\bbE[\Delta^V]=0$, the ex-ante profit of the insider is given by
\be \label{e:xanteIp}
\frac{1}{2\lambda }\left(P_0-\mu\right)^2+\frac{\gamma^2}{2\lambda} + \frac{\sigma^2}{2}\lambda.
\ee
%
%
\begin{table}[tb]
	\caption{Expected profit of insider/strategic trader}
	\label{T:profit}\vspace{1mm}
	\par
	\begin{center}
		
		\bgroup
		\def\arraystretch{1.25}
		{
			\begin{tabular}{cccc}
				& Competitive Agents & Market Makers & Additional value \\
				\hline
				Insider trader & $\gamma \sigma \sqrt{1+\frac{\rho_M^2}{4}}$	& $\gamma \sigma \pi (\rho_M)$ & $\gamma \sigma \Delta(\rho_M)$\\
				\hline
				Strategic trader & $\gamma\sigma\frac{\rho_M}{4}$&  $\gamma\sigma\frac{\rho_M}{3}$&  $\gamma\sigma\frac{\rho_M}{12}$  \\
				\hline
				Normalized value of information & $\sqrt{1+\frac{\rho_M^2}{4}}-\frac{\rho_M}{4}$ & $v(\rho_M)$& \\
				\hline
			\end{tabular}
		}
		\egroup
		
	\end{center}
	\par
	\begin{spacing}{1.0}
			\footnotesize In above $\Delta(x):=\frac{\Delta_0(u(x))}{2x}, \,\Delta_0(x):=x\log x -x +1, \, u(x):=1- \frac{2}{1+\sqrt{1+\frac{4}{x^2}}}$, $v(x):=\frac{u(x)\log u(x)+\frac{1+7 u(x)-8u^2(x)}{3 u(x)}}{2x}$, and $\pi(x):=\frac{x}{3}+v(x)$. 
	\end{spacing}\vspace{2mm}
\end{table}
This profit consists of three parts,  first of which disappears when there is no mispricing by the liquidity providers, i.e. when $P_0=\bbE[V]$. We will identify this part of the profit as the gains due to the initial mispricing by the market makers resulting from their competition.  To understand the role of the second component, i.e. $\frac{\gamma^2}{2\lambda }$, let's first consider the market with competitive agents. In this case, there is no mispricing and the anticipated profit (\ref{e:xpostIp}) upon receiving the signal that corresponds to this part of the ex-ante profit is given by $\frac{(\Delta^V)^2}{2\lambda}$. When $V=0$, i.e. the time-$0$ expected value  of the asset by the competitive agents and the insider coincide, this part of the anticipated profit vanishes, justifying the identification of the expected value of this part as {\em purely informational profit}. Similarly, in the market with market makers the term $\frac{\gamma^2}{2\lambda}$ is identified as purely informational profit, too.

The profit of the strategic trader in market makers' equilibrium,
\be \label{e:xprftstr}
\frac{1}{2 \lambda_0}(P_0-\mu)^2+ \frac{\lambda_0\sigma^2}{2},
\ee
 has two components with the first one being due to the mispricing only and the second one due to providing insurance to the market makers. Thus, comparison of the latter to the third component of the insider's ex-ante profit reveals that the third part of the insider's profit is generated by providing insurance to the risk averse liquidity suppliers.  The difference between the ex-ante profit of the insider and that of the strategic trader is the value of information. This value of information normalized by its Kyle counterpart, i.e. $\gamma\sigma$, is reported in Table \ref{T:profit} and depicted in Figure \ref{fig:profit}. 
 
It is apparent from Figure \ref{fig:profit} that  for reasonable values of market adjusted risk aversion parameter (i.e. $\rho_M\in [0,1)$), the value of information is smaller when the liquidity providers are risk averse. On the other hand, for  large values of $\rho_M$ the value of information in a market with risk averse liquidity providers becomes much higher then that  in the Kyle's model, with the ratio going to infinity in the limit. Moreover, further examination of the Figure \ref{fig:profit} and Table \ref{T:profit} reveals\footnote{Note that normalization doesn't involve $\rho$ and therefore doesn't change the shape of the value of information as a function of $\rho$.} that there is a value of the risk aversion parameter $\rho$ that minimizes\footnote{This value will be outside the reasonable range for risk aversion as it is larger than 10 (20) for competitive agents (market makers) equilibrium.} the value of information and, therefore, reduces the incentive for the insider to participate in the market.

Figure \ref{fig:profit} also reveals that  although the normalized value of information is  decreasing within the reasonable range of market adjusted risk aversion, it is non-monotone in general . The reason for that is the fact that the purely informational profit and the excess insurance profit  react in opposite directions to changes in $\rho_M$.  Indeed, as the risk aversion increases the market makers are willing to pay more for an insurance against large fluctuations in the demand. Moreover, in the presence of adverse selection this effect is magnified, which entails that $\lambda$ is increasing faster than $\lambda_0$ as is evident in Table \ref{T:depth}. On the other hand, the purely informational component is decreasing with an increase in the market adjusted risk aversion since, in the presence of strong mean reversion, it becomes more costly for the insider to drive the price to its fundamental value.
\begin{figure}[h]
	\begin{center}
		$
		\ba{cc}
		\scalebox{0.27}{\includegraphics{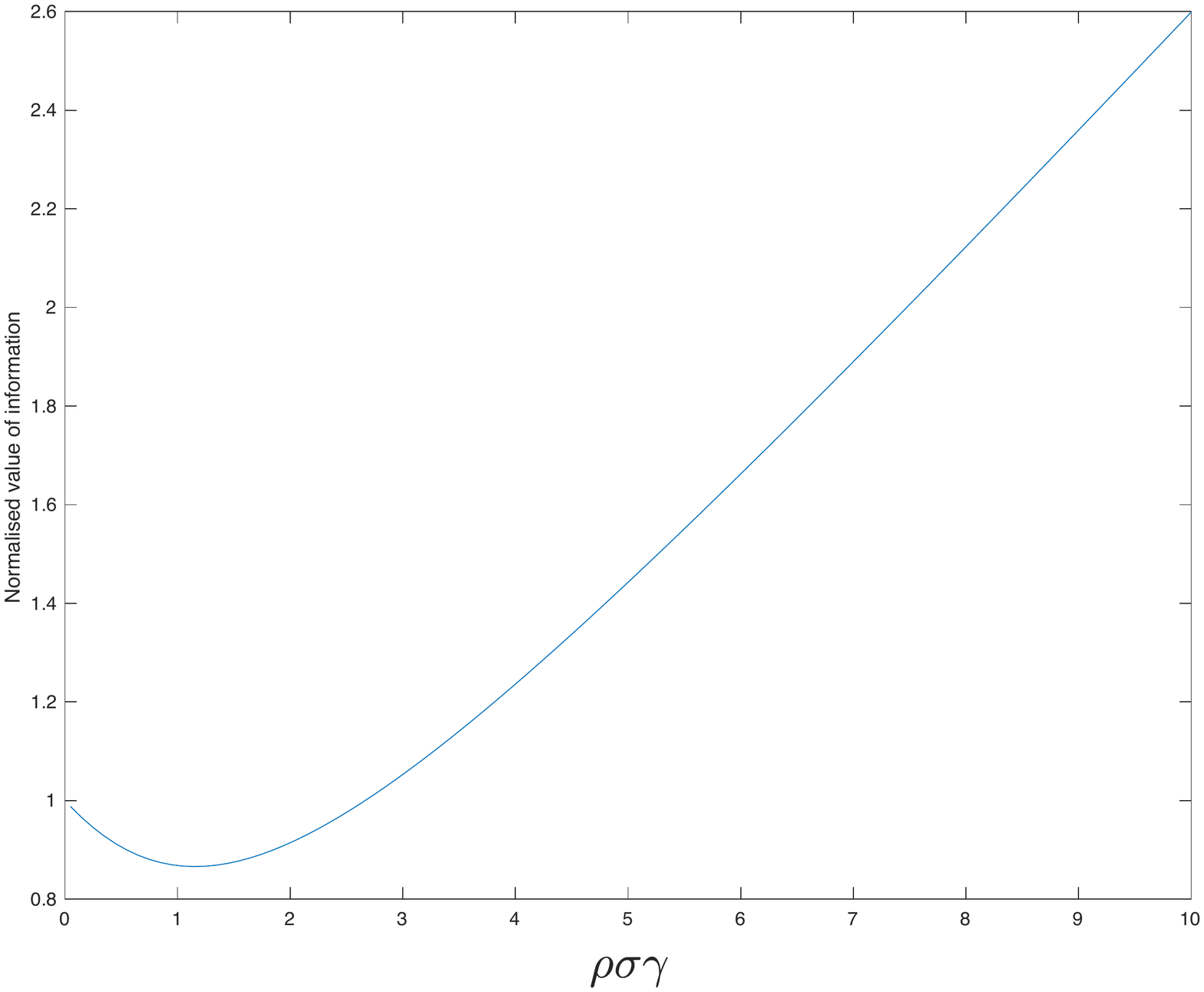}} &\scalebox{0.27}{\includegraphics{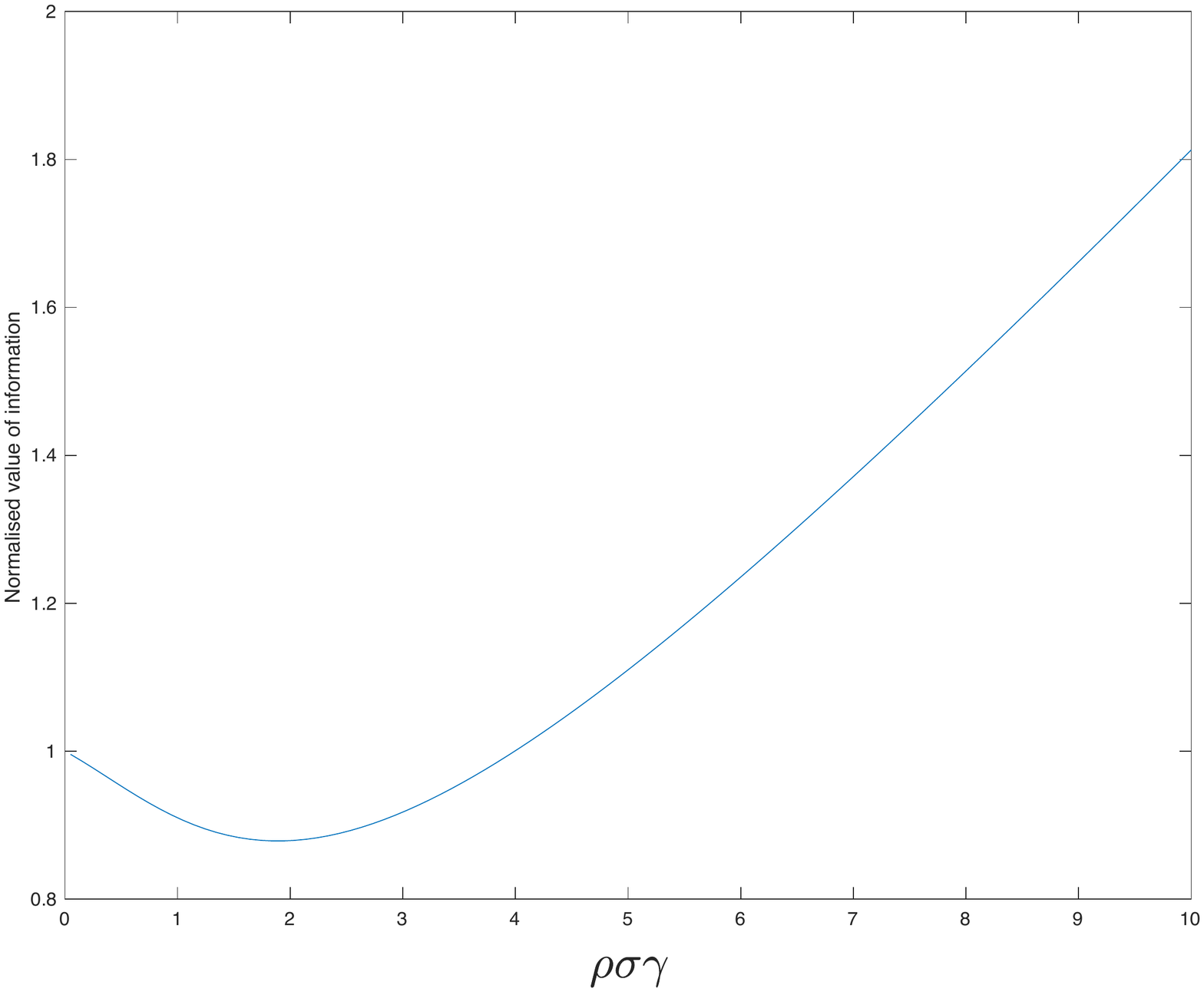}}
		\ea
		$
	\end{center}
	\caption{Normalized value of information in competitive agents equilibrium is reported in the left pane and the right plot illustrates corresponding value for  the market makers equilibrium.}
	\label{fig:profit}
\end{figure}

The difference in the normalized value of information between the markets with market makers and competitive agents displays a similar non-monotone pattern as illustrated by Figure \ref{fig:xvalinf}.  As one can observe, the information is less valuable when the liquidity is supplied by competitive agents  in the reasonable range of  market adjusted risk aversion . This implies that strategic traders in markets with market makers  have more incentives to acquire private information under typical market conditions. 
\begin{figure}[h]
	\centering
	\includegraphics[height=6cm, width=10cm]{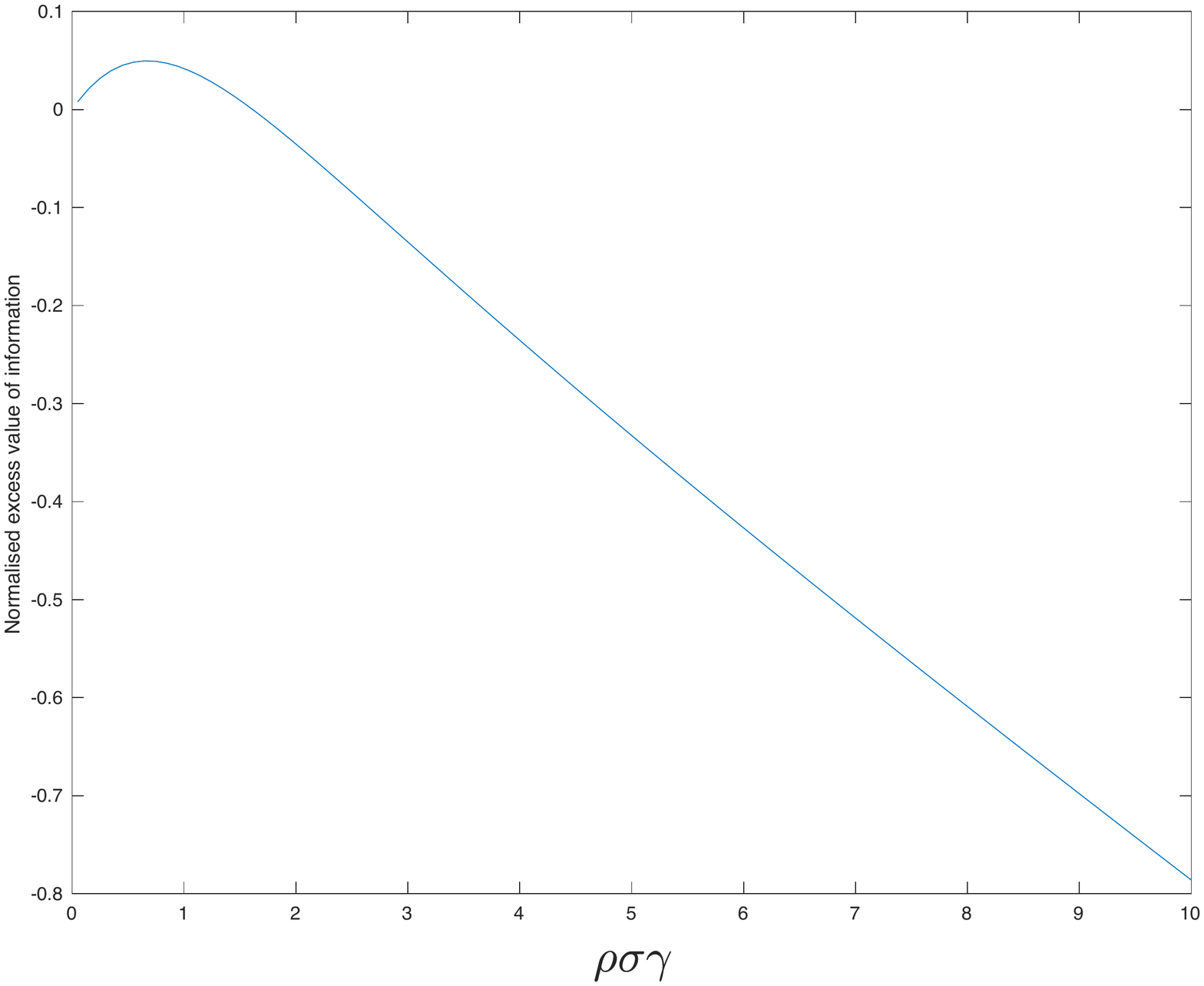}
	\caption{The difference between the normalized value of information in the market makers and competitive agents equilibrium as a function of $\rho_M$.}
	\label{fig:xvalinf}
\end{figure}
Despite this nonlinear pattern, an insider who already possesses private information chooses to trade with market makers as direct calculations demonstrate that the function $\Delta$ from Table \ref{T:profit} is positive. The same table also shows that strategic traders prefer to trade with market makers as well. The noise traders are indifferent between the two types of liquidity provision.  Indeed, their expected profit is given by 
\[
\bbE[\sigma \int_0^TB_t dP_t]=\sigma \bbE[B_1 V-\int_0^1P_t dB_t -\lambda \sigma]=-\lambda\sigma^2,
\]
which is identical across the markets. This in particular implies that the profits of the competitive agents above the zero-utility level are passed to the strategic traders in market makers equilibrium.  Thus, both informed and uninformed strategic traders enjoy larger profits with market makers. It is easy to see from Table \ref{T:profit}  that  the uninformed strategic trader's profit increases by $\frac{1}{3}$ when trading with market makers. Moreover, direct  calculations show that 
\[
\frac{1}{3} >\frac{\Delta(\rho_M)}{\sqrt{1+\frac{\rho^2_M}{4}}}
\]
 for all $\rho_M>0$, which implies  that the increase in profit is higher for uninformed strategic traders.  Finally, since all the traders prefer liquidity provided by the market makers, the trading desks executing client orders should optimally choose market makers over competitive agents. }}

\section{Conclusion} \label{s:C}
We have investigated the role of market mechanism in liquidity and trader profits.   Identification of a universal parameter, $\rho_M$, allowed us to compare equilibrium outcomes across different markets. The resulting equilibria let us conclude that the retail agents (informed, uninformed and noisy) would prefer their order routed to a market maker.  

The risk aversion of liquidity providers alters the equilibria not only quantitatively but also qualitatively.  The most significant distinction from the earlier models is that, as a result of risk sharing between the insider and the market makers,  the demand process becomes mean reverting although the cumulative noise demand is postulated to be a Brownian motion.

This mean reversion is responsible for the other departures from the previous literature that we highlighted in this paper. In particular, it entails conspicuous insider trading, and price reversal.

We also studied the sensitivity of market parameters to risk aversion. We confirm the intuition that  the market depth  decreases, prices exhibit stronger reversal, become less efficient and converge to the fundamental value at a lower rate as the liquidity providers become more risk averse. An unforeseen consequence of our models is the non-monotone relationship between the market adjusted  risk aversion and the value of private information. In particular, we observe that there exists a critical risk aversion level at which the amount a strategic trader is willing to pay for the private information is minimal.


\bibliographystyle{plainnat}
\bibliography{ref.bib}
\appendix
 \section{Proofs of main theorems}
\begin{proof}[Proof of Theorem \ref{t:eqLP}]
	We will start with proving that $P^* \in \cP(X^*)$.  Note that
	\be \label{e:X*lp}
	dX^*_t= \left\{-\sigma^2 \frac{\lambda^*\rho}{1+ \lambda \rho\sigma^2(1-t)}(X^*_t+Z_t)+\sigma^2 \frac{p_x(t,X^*_t+Z_t;1,\frac{V-\mu}{\lambda^*})}{p(t,X^*_t+Z_t;1,\frac{V-\mu}{\lambda^*})}\right\}dt,
	\ee
	where $p$ is the transition density of the Ornstein-Uhlenbeck process $R$ defined by
	\be \label{e:Rsde}
	dR_t=\sigma d\beta_t -\sigma^2 \frac{\lambda^*\rho}{1+ \lambda^* \rho\sigma^2(1-t)}R_t dt,
	\ee
	where $\beta$ is a Brownian motion. Thus, it follows from Theorem 2.2 in \cite{CDbridge} that $X^*+Z$, in its own filtration, is a weak solution of (\ref{e:Rsde}) since $R_1$ has the same distribution as $\frac{V-\mu}{\lambda}$.  Therefore,  
	\be \label{e:P*projection}
	dP^*_t=\lambda^* \left(\sigma d\beta^*_t  -\sigma^2 \frac{\rho(P^*_t-\mu)}{1+ \lambda^* \rho\sigma^2(1-t)} dt\right),
	\ee
	and $\beta^*$ is an $\cF^M$-Brownian motion, which clearly admits a unique martingale measure.  
	
	Moreover, Theorem 2.2 in \cite{CDbridge} also yields that $X^*+Z$ converges a.s. to $\frac{V-\mu}{\lambda^*}$. Thus, in view of Proposition \ref{p:ioptimality},  $X^*$ is optimal provided it is admissible.  Indeed, using the transition density of the bridge process provided by Theorem 2.2 in \cite{CDbridge}, it can be directly verified that
	\[
	\bbE^v\int_0^1 (\lambda^* (X^*_t+Z_t))^2 dt <\infty.
	\]
	Thus, $X^*$ is admissible, and, therefore optimal. Moreover, the value function reads 
	\[
	\psi(0,0)=\frac{1}{2 \lambda^*}(P_0-V)^2+ \frac{\lambda^*\sigma^2}{2}.
	\]
	Taking expectation yields ex-ante profit as claimed.

	We next turn to the optimisation problem for competitive agents, which is given by 
	\[
	\sup_{\theta \in \cA(0,\lambda^*,X^*)}\bbE\left[1-\exp\left\{-\rho\left(\int_0^1\theta_tdP^*_t \right)\right\} \right],
	\]
	where $P^*$ is given by  (\ref{e:P*projection}) and $P^*_0=\mu$.
	
	Let us introduce $Y^*:=X^*+Z$ and observe that
	\[
	dY^*_t=\sigma d\beta^*_t  -\sigma^2 \frac{\lambda^* \rho Y^*_t}{1+ \lambda^* \rho\sigma^2(1-t)} dt.
	\]
	In view of Theorem 2.1 in \cite{DGRSSS02} the optimal final wealth of  the competitive agent is given by
	\[
	\frac{1}{\rho}\left(\bbE^{\bbQ}[\log L_1]-\log L_1\right),
	\]
	where 
	\[
	L_1=\exp\left(\int_0^1  \frac{\rho\lambda^*Y^*_t }{1+ \lambda^* \rho\sigma^2(1-t)}dY^*_t +\frac{\sigma^2\rho^2 (\lambda^*)^2}{2}\int_0^1  \left(\frac{Y^*_t}{1+ \lambda^* \rho\sigma^2(1-t)}\right)^2dt\right)
	\]
	and $\bbQ$ is the equivalent measure defined by $\frac{d\bbQ}{d\bbP}=L_1$. Note that  $Y^*$ is a $\bbQ$-martingale with quadratic variation $d[Y^*_t,Y^*_t]=\sigma^2dt$.  Moreover, integration by parts applied to 
	\[
	\frac{(Y_t^*)^2}{1+\lambda^*\rho \sigma^2(1-t)}
	\]
	yields
	\[
	\log L_1=\frac{\lambda^*\rho}{2}\left( (Y^*_1)^2-\log(1+\lambda^*\rho\sigma^2))\right).
	\]
	
	This in turn implies
	\[
	\bbE^{\bbQ}[\log L_1]=\frac{\sigma^2}{2}\int_0^1  \frac{\rho^2(\lambda^*)^2\sigma^2 t}{(1+ \lambda^* \rho\sigma^2(1-t))^2}dt=\frac{1}{2}\left(\lambda^*\rho\sigma^2-\log (1+\lambda^*\rho\sigma^2)\right).
	\]
	Therefore, the final optimal wealth of the competitive agent is given by
	\[
	\frac{\lambda^*}{2}(\sigma^2-(Y^*_1)^2).
	\]
	This proves the optimality of $\theta^*$ since $\int_0^1 \theta^*_t dP^*_t= \frac{\lambda^*}{2}(\sigma^2-(Y^*_1)^2)$.
	
	Moreover, the expected utility is given by
	\[
	\bbE\left[1-\exp\left(-\bbE^{\bbQ}[\log L_1]+\log L_1\right)\right]=1-e^{-\bbE^{\bbQ}[\log L_1]}=1-\sqrt{1+\lambda^*\rho\sigma^2}e^{-\frac{\lambda^*\rho\sigma^2}{2}}.
	\]
	Using the identity $\lambda^*\rho\gamma^2=(\lambda^*)^2-\frac{\gamma^2}{\sigma^2}$ we arrive at (\ref{e:lputility}).
	
	Since it is obvious that the market clears, $((m^*,\lambda^*),\theta^*=(\theta^{*,a})_{a\in [0,1]},X^*)$ is an equilibrium. Finally, the equilibrium price process follows (\ref{e:priceLP}) due to (\ref{e:X*lp}) and the explicit form of $p$.
	
\end{proof}
\begin{proof}[Proof of Theorem \ref{t:benchLP}]
	Observe that the pricing rule in terms of the notation of Proposition \ref{p:ioptimality} can be expressed as $\lambda(t)=\lambda^*, \, t<1$, $\lambda(1)=2\lambda^*$, $\phi=\mu$, $m_0=\mu$ and $c_0=2\lambda^*$. Thus, Condition C2 is satisfied with $2\lambda(1)\lambda(1-)=c_0^2$. Moreover, since $\Delta Y^*$ satisfies \eqref{e:dY1opt}, it suffices to show the admissibility to prove the optimality of the stated strategy. This, however, follows from the fact that $Y^*$ is Gaussian and $P^*$ is an affine function of $Y^*$.
	
	We next turn to the optimisation problem for competitive agents, which is given by 
	\[
	\sup_{\theta \in \cA(0,\lambda^*,X^*)}\bbE\left[1-\exp\left\{-\rho\left(\int_0^{1}\theta_{t-}dP^*_t + \theta_1(V-P_1^*)\right)\right\} \right].
	\]
	Since all stochastic processes involved in the above expression are independent of $V$, the above reduces to
	\[
	\sup_{\theta \in \cA(0,\lambda^*,X^*)}\bbE\left[1-\exp\left\{-\rho\left(\int_0^{1}\theta_{t-}dP^*_t -2\lambda^*Y_1^* \theta_1-\frac{\rho\gamma^2}{2}\theta_1^2\right)\right\} \right].
	\]
	Maximising over $\theta_1$ yields that $\theta_1^*=-Y^*_1$. Thus, the optimisation problem becomes
	\[
	\begin{split}
		&\sup_{\theta \in \cA(0,\lambda^*,X^*)}\bbE\left[1-\exp\left\{-\rho\left(\int_0^{1}\theta_{t-}dP^*_t +\lambda^*(Y_1^*)^2 \right)\right\} \right]\\
		&=\sup_{\theta \in \cA(0,\lambda^*,X^*)}\bbE\left[1-\exp\left\{-\rho\left(\int_0^{1}\theta_{t-}dP^*_t +\frac{\lambda^*}{4}(Y_{1-}^*)^2 \right)\right\} \right]
	\end{split}	
	\]
	Next, consider the process
	\[
	N_t:=\exp(-\alpha_2(t)(Y^*_t)^2 -\alpha_0(t)),
	\]
	where
	\[
	\begin{split}
		\alpha_2(t)&:=-\frac{a(t)}{2}- \frac{\lambda^*\rho}{4+2\lambda^*\rho\sigma^2(1-t)}\\
		a(t)&:=-\frac{\lambda^*\rho}{1+\lambda^*\rho\sigma^2(1-t)}\\
		\alpha_0(t)&=\sigma^2 \int_t^1\alpha_2(s)ds.
	\end{split}
	\]	
	Direct calculations show that 
	\[
	dN_t=-2\alpha_2(t)Y_t \sigma dB_t.
	\]
	Thus, under $\tilde{\bbP}$ defined by $\frac{d\tilde{\bbP}}{d\bbP}=\sqrt{\frac{1+\lambda^*\rho \sigma^2}{1+\half\lambda^*\rho \sigma^2}}N_1$, $Y^*$ follows
	\[
	dY^*_t= \sigma d\tilde{B}_t  +\sigma^2 a(t)Y^*_t dt, \quad t<1,
	\]
	where $\tilde{B}$ is a $\tilde{\bbP}$-Brownian motion. Consequently, the optimisation problem can be cast under $\tilde{\bbP}$ as follows
	\[
	\sup_{\theta \in \cA(0,\lambda^*,X^*)}\tilde{\bbE}\left[1-\sqrt{\frac{1+\half\lambda^*\rho \sigma^2}{1+\lambda^*\rho \sigma^2}}\exp\left\{-\rho\left(\int_0^{1-}\theta_{t}dP^*_t \right)\right\} \right]
	\]
	In view of Theorem 2.1 in \cite{DGRSSS02} the optimal final wealth of  the competitive agent is given by
	\[
	\frac{1}{\rho}\left(\bbE^{\bbQ}[\log L_1]-\log L_1\right),
	\]
	where 
	\[
	L_1=\exp\left(\int_0^1  \frac{\rho\lambda^*Y^*_t }{1+ \lambda^* \rho\sigma^2(1-t)}\sigma d\tilde{B}_t -\frac{\sigma^2\rho^2 (\lambda^*)^2}{2}\int_0^1  \left(\frac{Y^*_t}{1+ \lambda^* \rho\sigma^2(1-t)}\right)^2dt\right)
	\]
	and $\bbQ$ is the equivalent measure defined by $\frac{d\bbQ}{d\tilde{\bbP}}=L_1$. Note that  $Y^*$ is a $\bbQ$-martingale with quadratic variation $d[Y^*_t,Y^*_t]=\sigma^2dt$.  Moreover, integration by parts applied to 
	\[
	\frac{(Y_t^*)^2}{1+\lambda^*\rho \sigma^2(1-t)}
	\]
	yields
	\[
	\log L_1=\frac{\lambda^*\rho}{2}\left( (Y^*_{1-})^2-\log(1+\lambda^*\rho\sigma^2))\right).
	\]
	
	This in turn implies
	\[
	\bbE^{\bbQ}[\log L_1]=\frac{\sigma^2}{2}\int_0^1  \frac{\rho^2(\lambda^*)^2\sigma^2 t}{(1+ \lambda^* \rho\sigma^2(1-t))^2}dt=\frac{1}{2}\left(\lambda^*\rho\sigma^2-\log (1+\lambda^*\rho\sigma^2)\right).
	\]
	Therefore, the final optimal wealth of the competitive agent is given by
	\[
	\frac{\lambda^*}{2}(\sigma^2-(Y^*_{1-})^2).
	\]
	This proves the optimality of $\theta^*$ since $\int_0^{1-} \theta^*_t dP^*_t= \frac{\lambda^*}{2}(\sigma^2-(Y^*_{1-})^2)$.
	
	Moreover, in view of $\tilde{E}[L_1]=1$, the expected utility is therefore given by
	\[
	\tilde{\bbE}\left[1-\sqrt{\frac{1+\half\lambda^*\rho \sigma^2}{1+\lambda^*\rho \sigma^2}}\exp\left(-\bbE^{\bbQ}[\log L_1]+\log L_1\right)\right]=1-\sqrt{1+\half \lambda^*\rho\sigma^2}e^{-\frac{\lambda^*\rho\sigma^2}{2}},
	\]
	which proves (\ref{e:lputilitys}).
	
	Since it is obvious that the market clears, $((m^*,\lambda^*),\theta^*=(\theta^{*,a})_{a\in [0,1]},X^*)$ is an equilibrium.
	%
\end{proof}
{
\begin{proof}[Proof of Theorem \ref{t:eqMM}]
	It follows from Theorem 2.2 in \cite{CDbridge} that $X^*+Z$, in its own filtration, is a weak solution of (\ref{e:Ri}) since $R_1^{(i)}$ has the same distribution as $\frac{V-\phi_i^*}{\lambda^*}$. In particular,
	\be \label{e:Y*mm}
	dY^*_t=\sigma^2(a(t)Y_t^* + b(t))dt +\sigma d\beta^*_t
	\ee
	for some Brownian motion $\beta^*$. This readily proves (\ref{e:demandMMm}) and (\ref{e:priceMMm}).
	
	Moreover, Theorem 2.2 in \cite{CDbridge} also yields that $X^*+Z$ converges a.s. to $\frac{V-\phi_i^*}{\lambda^*}$. Thus, in view of Proposition \ref{p:ioptimality},  $X^*$ is optimal provided it is admissible.  Indeed, using the transition density of the bridge process provided by Theorem 2.2 in \cite{CDbridge}, it can be directly verified that
	\[
	\bbE^v\int_0^1 (\lambda^* (X^*_t+Z_t))^2 dt <\infty.
	\]
	Thus, $X^*$ is admissible, and, therefore optimal. Moreover, the value function given by Proposition \ref{p:ioptimality} reads 
	\[
	\psi(0,0)=\frac{1}{2 \lambda^*}(P_0-V)^2+ \frac{\lambda^*\sigma^2}{2}.
	\]
	Taking expectation yields ex-ante profit as claimed.

	Next  define 
	\[
	G_t:=-\int_0^{t} Y^{\ast}_s\,dP^{\ast}_s +Y^{\ast}_{t} (P^{\ast}_{t}-V)
	\]
	and note that for $t<1$, $G_t$ is the final wealth of an agent who got the allocation at time $0$ and decided to quit market making at time $t$. Moreover, since $\Delta Y^*_1=0$, $G_1$ is the wealth of a market maker who follows the stated strategy until time $1$ and wins the allocation at time $0$ when all market makers quote $\lambda^*$. 
	
	Note that the  above in conjunction with (\ref{e:Y*mm}) yield that if all market makers quote $\lambda^*$, their utility  is $0$. Indeed, since $P^*_t= \lambda^* Y^*_t + \phi^*$ and $P^*_1=V$,
	\bean
	M_t:=1-\bbE_t[U(G_{t})]&=&\bbE_t\left[\exp\left(\rho \lambda^* \int_0^{t} Y^{\ast}_s\,dY^{\ast}_s -\rho\lambda^*Y^{\ast}_{t} (Y^{\ast}_{t}-Y^*_1)\right)\right]\\
	&=&\exp\left(-\frac{\rho \lambda^*}{2}((Y^*_t)^2 +\sigma^2 t)\right)\bbE_t\left[\exp\left(\rho \lambda^* Y^{\ast}_{t}Y^*_1)\right)\right]\\
	&=&\exp\left(-\frac{\rho \lambda^*\sigma^2 t}{2}+\alpha_1(t)Y_t +\alpha_2(t)Y_t^2\right),
	\eean
	where
	\bean
	\alpha_1(t)&:=&\lambda^*\rho \sigma^2 e^{\sigma^2 A(1)}\int_t^1e^{-\sigma^2 A(s)}b(s)ds, \mbox{ and }\\
	\alpha_2(t)&:=& \lambda^*\rho e^{\sigma^2(A(1)-A(t))}+ \half (\lambda^*)^2\rho^2\sigma^2 \int_t^1 e^{2\sigma^2(A(1)-A(s))}ds- \frac{\lambda^* \rho }{2},
	\eean
	and $A(t):=\int_0^t a(s)ds$ using the moment generating function of the normal distribution associated with the conditional distribution of $Y^*_1$ given $Y^*_t$. By Lemma \ref{l:martingaleU} $M$ is a martingale. Therefore, $\bbE[U(G_1)]=1-M_0=0$. Since $U(0)=0$, the ex-ante utility of an arbitrary market maker is given by $\frac{1}{N}\bbE[U(G_1)]=0$ when all market makers quote $\lambda^*$. 

	Observe that Proposition \ref{p:ioptimality} implies that the market $\lambda$ must be an increasing function of time on $[0,1)$.  Moreover, quoting a $\lambda$ below $\frac{\lambda^*}{2}$ at time $1$ leads to infinite loss and, therefore, is  suboptimal.  In conjunction with our conventions on time priority this yields only two options for the deviating market maker: i) quote a higher $\lambda$ starting from time $s\geq 0$, which amounts to stopping making the market on $[s,1]$, or ii) quote a smaller $\lambda$ at time $0$ with the possibility of increasing it at later dates and opportunity to decrease at time $1$.  Let $(\lambda, \phi)$ denote the market quote resulting from the deviation. 
	
Consider Case i). In this case we can assume without loss of generality that deviating market maker $j$ quotes
	\[
	\lambda^j(t)=\lambda^* \chf_{[t<\tau^j]}+ (\lambda^*+1)\chf_{[\tau^j\leq t\leq 1]}.
	\]
Observe that the market quote resulting from this deviation is  still the equilibrium market quote $(\lambda^*,\phi^*)$.  Thus, by our convention, the insider follows the same equilibrium strategy in this setting as well. In particular $\Delta Y_1=0$ and $P_{1-}=V$. Thus, the final wealth of the market maker is given by
	\[
	G^j_1=- \int_0^{\tau^j}Y^*_tdP^*_t +Y^*_{\tau^j}(P^*_{\tau^j}-V)=G_{\tau^j}.
	\]
	Since $M$ is a martingale, it follows that $ \bbE[U(G^j_1)]=0$.

	 If the deviating market maker chooses to undercut at time $0$, then he has the option to increase $\lambda$ on $[0,1)$,  can quit the market at $t \in (0,1)$ by quoting $\lambda^*$ or above, and can decide to meet the bulk order at $t=1$.   Let $(\lambda, \phi)$ denote the market quote resulting from such a deviation and $\hat{G}_1$ be the gains of a fictitious market maker as defined in Remark \ref{r:fictitiousgain}. Suppose that the market maker quits making the market at $t^*\in (0,1]$ with $t^*=1$ and $\Delta Y^j_1 \neq 0$ meaning that he makes the market until the end. Denoting his gains process by $G^j$ we have
	 \[
	 \begin{split}
	 	G^j_1=&\hat{G}_1-\frac{\lambda(1-)\sigma^2}{2}(1-t^*)-\frac{(P_{1-}-V)^2}{2\lambda(1-)}\chf_{[t^*<1]} +\Delta Y_1 (V-P_1)\chf_{[\Delta Y^j_1 =0]}\\
	 	&=-\frac{(P_0-V)^2}{2\lambda(0)}+\frac{\sigma^2}{2}\int_0^{t^*}\lambda(s)ds-\Delta Y_1 (V-P_1)\chf_{[\Delta Y^j_1 \neq 0]}+\frac{(P_{1-}-V)^2}{2\lambda(1-)}\chf_{[t^*=1]} \\
	 	&\leq -\frac{(P_0-V)^2}{2\lambda(0)}+\frac{\sigma^2}{2}\lambda^*,
	 \end{split}
	 \]
where the second equality follows from \eqref{e:BGIf} and the last one follows from Remark \ref{r:DMM1} when $t^*=1$ and $\Delta Y^j_1 \neq 0$. In case $t^*=1$ and $\Delta Y^j_1 =0$, $\lambda(1)>\lambda(1-)/2$ and, therefore $P_{1-}=V$ in view of Proposition \ref{p:ioptimality}.  Observe that the price priority rule implies that $\Delta Y^j_1 =0$ if $t^*<1$. This shows $\bbE[U(G^j_1)]< 0$.
	 
	 Hence,  $((\Lambda^*, \Phi_i^*),X^*)$ is an equilibrium. 
	
	To conclude, observe that the dynamics given by (\ref{e:demandMMi}) and (\ref{e:priceMMi}) follow from the fact that $Y^*=X^*+\sigma B$ and the explicit form of the transition density of $R$.
\end{proof}

\begin{proof}[Proof of Theorem \ref{t:benchMM}]
	In view of Proposition \ref{p:ioptimality},  $X^*$ is optimal provided it is admissible. Since the price process is Gaussian, (\ref{a:Xintegrability}) clearly holds. Moreover, the value function reads 
	\[
	\psi(0,0)=\frac{1}{2 \lambda^*(0)}(P_0-\mu)^2+ \frac{\lambda^*(0)\sigma^2}{2}=\frac{\rho\gamma^2\sigma^2}{4}.
	\]
	
Next recall that for $t<1$
\[
G_t=-\int_0^{t} Y^{\ast}_s\,dP^{\ast}_s +Y^{\ast}_{t} (P^{\ast}_{t}-V)
\]
is the final wealth of a market maker who got the allocation at time $0$ and decided to quit market making at time $t$ as in the proof of Theorem \ref{t:eqMM}. Note that since in equilibrium $Y$ has a jump at $1$,  the final   wealth of a market maker who follows the stated strategy until time $1$ and wins the allocation at time $0$ would be different. To be more precise
\[
G_1=-\int_0^{1-} Y^{\ast}_s\,dP^{\ast}_s +Y^{\ast}_{1-} (P^{\ast}_{1-}-V)+\Delta Y^{\ast}_1(P^{\ast}_{1}-V).
\]

Note that the  above yields that if all market makers quote $\lambda^*$, their utility  is $0$. Indeed, since $P^*_t= \lambda^*(0) Y^*_t + \phi^*$ for $t<1$, we can define $M$ via
\bean
M_t:=1-\bbE_t[U(G_{t})]&=&\bbE_t\left[\exp\left(\rho \lambda^*(0) \int_0^{t} Y^{\ast}_s\,dY^{\ast}_s -\rho Y^{\ast}_{t} (\lambda^*(0)Y^{\ast}_{t}+\phi^*-V)\right)\right]\\
&=&\exp\left(-\frac{\rho^2\gamma^2\sigma^2 t}{4}+\rho(\mu-\phi^*)Y^*_t +\frac{\rho^2\gamma^2}{4}(Y_t^*)^2\right),
\eean
where the last equality follows from the independence of $V$ and $Y^*$ and the expression for $\lambda^*$. Direct application of Ito's formula yields that $M$ is a martingale on $[0,1)$. Furthermore, defining 
\[
M_1:=1-\bbE_1[U(G_{1})]
\]
one can directly verify that 
\be \label{e:StrUtil}
M_1=\exp\left(V\rho\frac{\mu-\phi^*}{\lambda^*(0)}-\frac{\sigma^2\rho\lambda^*(0)}{2}-\rho \frac{\mu^2 -(\phi^*)^2}{2\lambda^*(0)}\right)=1-U\left(-\frac{(\mu-\phi^*)^2}{2\lambda^*(0)}+\frac{\sigma^2\lambda^*(0)}{2}+
\frac{\mu-\phi^*}{\lambda^*(0)}(\mu-V)\right).
\ee
Thus, $\bbE[M_{1}]=1$. Since $U(0)=0$, utility of any market maker is given by $\frac{1}{N}\bbE[U(G_1)]=0$. 
	
	Observe that Proposition \ref{p:ioptimality} implies that the market $\lambda$ must be an increasing function of time on $[0,1)$ and 
	\[
	2\lambda(1) \lambda(1-)\geq 1.
	\]
This implies that a deviating market maker who doesn't make the market on $[0,1)$ cannot submit a winning quote at time $1$ as undercutting at time $1$ would cause infinite loss according to Proposition \ref{p:ioptimality}. In conjunction with our conventions on time priority this yields only two options for the deviating market maker: i) quote a higher $\lambda$ starting from time $s\geq 0$, which amounts to stopping making the market on $[s,1]$, or ii) quote a smaller $\lambda$ at time $0$ with the possibility of increasing it at later dates and an opportunity to decrease at time $1$.
	
	In the first case, if the market maker $j$  quotes a higher $\lambda$ starting at $s<1$, he will get no order on $[s,1]$ due to the no-shuffling rule and the fact that by undercutting market lambda at time $1$  will result in infinite loss.  Note that for this deviating strategy the resulting market quote does not change and therefore the equilibrium strategy of the strategic trader stated in the theorem is still optimal.  Since $M$ is a  martingale and the utility of the deviating market maker is given by $1- \bbE[M_s]$, such a deviation is suboptimal.

	If the deviating market maker chooses to undercut at time $0$, then he has two options: i) quit making the market at $t^* \in (0,1]$ by quoting the equilibrium $\lambda$ or above, and ii) make the market on $[0,1]$ by undercutting throughout. 
	
Let $(\lambda, \phi)$ denote the market quote resulting from such a deviation and $\hat{G}_1$ be the gains of a fictitious market maker as defined in Remark \ref{r:fictitiousgain}. Suppose that the market maker quits making the market at $t^*\in (0,1]$. Denoting his gains process by $G^j$,  we have via  Remark \ref{r:fictitiousgain} that 
	\[
	\begin{split}
		G^j_1=&\hat{G}_1-\frac{\lambda(1-)\sigma^2}{2}(1-t^*)+\frac{(\mu-V)^2-(P_{1-}-V)^2}{2\lambda(1-)}+\Delta Y_1 (V-P_1)\\
		&=-\frac{(P_0-\mu)^2}{2\lambda(0)}+\frac{\sigma^2}{2}\int_0^{t^*}\lambda(s)ds+ Y_{t^*}(\mu-V)\\
		&=-\frac{(\phi-\mu)^2}{2\lambda(0)}+\frac{\sigma^2}{2}\int_0^{t^*}\lambda(s)ds+ \frac{\mu-\phi}{\lambda(0)}(\mu-V),
	\end{split}
	\]
	where we use the fact that $P_{t^*}=\mu$ in view of Proposition \ref{p:ioptimality} since the strategic trader constructs a bridge at $t^*$.  Thus, direct comparison with \eqref{e:StrUtil} and the fact that $\lambda <\lambda^*$ on $[0,t^*)$ reveals that deviation is not profitable. 
	
	Moreover, if the deviating market maker improves the equilibrium quotes on $[0,1]$, his total gains are given by $\hat{G}_1$.  If $2\lambda(1)=\lambda(1-)$, Proposition \ref{p:ioptimality} implies $Y_1=\frac{\mu-\phi}{\lambda(1-)}$.  Otherwise, $Y_1=\frac{\mu-\phi}{\lambda(0)}$.  Also note that Remark \ref{r:DMM1} and \eqref{e:BGuf} yield that
	\[
	\hat{G}_1\leq -\frac{(\phi-\mu)^2}{2\lambda(0)}+\frac{\sigma^2}{2}\int_0^{1}\lambda(s)ds+Y_1(\mu-V).
	\]
	Therefore, another comparison with \eqref{e:StrUtil} and the fact that $\lambda <\lambda^*$ on $[0,1]$ leads to the claimed suboptimality. 

	Finally the dynamics of $Y^*$ follows from that  $Y^*_i=\sigma B + X^*_i$.
\end{proof}
}

\section{Exponential martingales of Ornstein-Uhlenbeck processes}
\begin{lemma}
	Suppose $Y$ solves
\[
dY_t=\sigma^2\left(a(t)Y_t +b(t)\right)dt + \sigma dB_t. \quad Y_0=0,
\]
for some continuously differentiable functions $a$ and $b$. Then,
\[
M_t:=\exp\left(\alpha_0(t)+\alpha_1(t)Y_t +\alpha_2(t)Y_t^2\right)
\]
is a martingale such that 
\bean
\alpha_0(1)&=&-\frac{\lambda \rho\sigma^2}{2},\\
\alpha_1(1)&=&0,\\
\alpha_2(1)&=&\frac{\lambda \rho}{2}
\eean
and $MY$ is a martingale
if and only if 
\bea
-2 \alpha_2(t)&=&a(t)=-\frac{\lambda \rho}{1+ \lambda\rho\sigma^2(1- t)},\\
b(t)&=&\alpha_1(t)=0,\\
\alpha_0(t)&=& -\frac{\lambda \rho\sigma^2 }{2}+\half\log(1+\lambda\rho\sigma^2(1-t)).
\eea
\end{lemma}
\begin{proof}
	Application of Ito's formula yields the following necessary and sufficient conditions on the coefficients $\alpha_i$ for $M$ to be a local martingale:
	\bea 
	2 \alpha_2(t) a(t)+ 2 \alpha_2^2(t)+\frac{\alpha_2'(t)}{\sigma^2}&=&0\label{e:ODE1};\\
	\frac{\alpha_1'(t)}{\sigma^2}+ \alpha_1(t)a(t)+2 \alpha_2(t)\alpha_1(t)+2 \alpha_2(t)b(t)&=&0;\label{e:ODE2}\\
	\frac{\alpha_0'(t)}{\sigma^2}+ \alpha_1(t)b(t)+ \alpha_2(t)+\frac{\alpha_1^2(t)}{2}&=&\label{e:ODE3}0.
	\eea
As a result
\[
dM_t=M_t (\alpha_1(t)+2\alpha_2(t)Y_t)\sigma dB_t.
\]	
Thus, in order for $MY$ to be a local martingale we must have
\[
\alpha_1(t)=-b(t)\; \mbox{ and }\; \alpha_2(t)=-\frac{a(t)}{2}.
\]
Plugging above into (\ref{e:ODE2}) yields $b'=-\sigma^2 a b$. Thus, $b\equiv 0$ since $b(1)=0$. Moreover, the remaining ODEs can be rewritten as
\bean
a'(t)&=& -\sigma^2 a^2(t);\\
\alpha_0'(t)&=&\frac{\sigma^2}{2}a(t),
\eean
which yields the claim in view of the boundary of conditions. Indeed, $Y$ is an Ornstein-Uhlenbeck process, which in turn implies that $M$ is a martingale by direct calculations. Moreover, Girsanov's theorem yields that $MY$ is a martingale as well.
\end{proof}
\begin{lemma} \label{l:martingaleU}
	Suppose $Y$ solves
	\[
	dY_t=\sigma^2\left(a(t)Y_t +b(t)\right)dt + \sigma dB_t. \quad Y_0=0,
	\]
	for some continuously differentiable functions $a$ and $b$. Then, $Y_t$ is normally distributed given $Y_s$ with mean
	\[
	Y_s e^{\sigma^2(A(t)-A(s))} +\int_s^t e^{\sigma^2(A(t)-A(r))}\sigma^2b(r)dr,
	\]
	and variance
	\[
	v^2(s,t):=\int_s^te^{2 \sigma^2(A(t)-A(r))}\sigma^2dr,
	\]
	where
	\[
	A(t):=\int_0^ta(s)ds.
	\]
	Moreover,
	\[
	M_t:=\exp\left(\alpha_0(t)+\alpha_1(t)Y_t +\alpha_2(t)Y_t^2\right)
	\]
	is a martingale, where
	\bean
	\alpha_0(t)&:=&-\frac{\lambda\rho\sigma^2}{2}t,\\
	\alpha_1(t)&:=&\lambda\rho \sigma^2 e^{\sigma^2 A(1)}\int_t^1e^{-\sigma^2 A(s)}b(s)ds, \mbox{ and }\\
	\alpha_2(t)&:=& \lambda\rho e^{\sigma^2(A(1)-A(t))}+ \half \lambda^2\rho^2\sigma^2 \int_t^1 e^{2\sigma^2(A(1)-A(s))}ds- \frac{\lambda \rho }{2},
	\eean
	and $A(t):=\int_0^t a(s)ds$, if and only if the following are satisfied:
	\bean
	a(t)&=&-\frac{\lambda\rho}{1+ \lambda \rho \sigma^2(1-t)} \mbox{ and }\\
	b(t)&=&\pm \frac{1}{2}\log(1+\lambda\rho\sigma^2(1-t))\sqrt{\frac{a(t)}{-\lambda\rho\sigma^2 (1-t)+\log (1+\lambda\rho\sigma^2(1-t))}}.
\eean
\end{lemma}
\begin{proof}
The first claim is well-known and follows from direct calculations.

We can identify $\alpha_2$ by using the moment generating function of a chi-squared random variable as follows:
Note that
\[
M_t:=\exp\left(\alpha_0(t)+\alpha_2(t)\left(Y_t+\frac{\alpha_1(t)}{2\alpha_2(t)}\right)^2-\frac{\alpha_1^2(t)}{4\alpha_2(t)}\right).
\]
By using the moment generating function of chi-squared r.v. and equating the coefficients of $Y_s^2$ in $M_s=\bbE[M_t|\cF_s]$, we find
\[
\alpha_2(t)e^{2\sigma^2(A(t)-A(s))}=\alpha_2(s)(1-2\alpha_2(t)v^2(s,t)).
\]
Next take $t=1$ in above and use the definition  for $\alpha_2$ to arrive at
\[
\phi^2(t)=\left[2\phi(t)+ \lambda\rho\sigma^2 \int_t^1\phi^2(u)du-1\right]\left(1-\lambda\rho\sigma^2 \int_t^1\phi^2(u)du\right),
\]
where 
\[
\phi(t)= e^{\sigma^2 (A(1)-A(t))}.
\]
The above is equivalent to
\[
\left(1-\phi(t)-\lambda \rho \sigma^2\int_t^1\phi^2(u)du\right)^2=0.
\]
Thus, using the fact that $\phi(1)=1$, we obtain
\[
\phi(t)=\frac{1}{1+\lambda \rho \sigma^2 (1-t)},
\]
 which implies
\[
a(t)=-\frac{\lambda \rho}{1+\lambda \rho \sigma^2 (1-t)}.
\]
Also observe that $\alpha_i s$ must also satisfy (\ref{e:ODE1})-(\ref{e:ODE3}) since $M$ is a local martingale. Note that (\ref{e:ODE2}) is redundant since $\alpha_2=-a/2$ by direct manipulation. Thus,  (\ref{e:ODE3}) simplifies to 
\[
-\frac{\rho\lambda}{2}+\alpha_1(t)b(t)-\frac{a(t)}{2}+\frac{\alpha_1^2(t)}{2}=0.
\]
By substituting $b=\frac{\alpha_1'}{\sigma^2 a}$ and $v=\alpha^2_1$,
we get the ODE 
\[
\frac{v'}{\sigma^2 a}+v=\lambda \rho + a,
\]
which has the following solution:
\[
v(t)=-a(t)\left(\lambda\rho\sigma^2 (1-t)-\log (1+\lambda\rho\sigma^2(1-t))\right).
\]
Note that $v$ is always positive since the term in the parentheses is decreasing to 0 at t=1.  

Calculating $b$ via $b=\frac{\alpha_1'}{\sigma^2 a}$ yields the claimed form of $b$.

Finally, since $Y$ is an Ornstein-Uhlenbeck process,  that $M$ is a martingale follows from direct calculations via the transition function. 
\end{proof}

{
\section{Strategic trader's optimality}
\begin{proposition} \label{p:ioptimality} Let  $\lambda$ be a right-continuous and piecewise constant function on $[0,1]$ and $\phi(t)=\phi + \chf_{[t=1]}\Delta \phi$. Consider the price process defined by
	\[
	P_t=\phi(t)+\int_0^t \lambda(s)\{dZ_s+dX_s\}
	\]
	and assume that  i) $\Delta \phi=c_1 P_{1-} +c_0$ for some  constant $c_1$, $c_0$. 
	 Denote the  trader's valuation of $V$ by $\tilde{V}$, i.e $\tilde{V}=V$ for the insider and $\tilde{V}=\mu$ for the uninformed strategic trader. 
	\begin{enumerate}
		\item The trader has finite profit if and only if the following conditions are satisfied:
		\begin{itemize}
			\item[C1] $\lambda$ is non-decreasing on $[0,1)$,  and $\frac{2 \lambda(1)}{\lambda(1-)}\geq (1+c_1)^2$. 
			\item[C2] If $ \frac{2 \lambda(1)}{\lambda(1-)}= (1+c_1)^2$, then  $c_0=-\mu c_1$ when $\tilde{V}=\mu$, and $c_0=c_1=0$,  when $\tilde{V}=V$. 
		\end{itemize}
		\item If the trader has finite profit,  any  admissible strategy is optimal if and only if
		\begin{align}
		 \Delta Y_1&=\frac{{\tilde{V}}-c_0-(1+c_1)P_{1-}}{2\lambda(1)},\, \mbox{a.s.}, \label{e:dY1opt}\\
		P_{1-}&=\frac{\tilde{V}((1+c_1)\lambda(1-)-2\lambda(1))-c_0(1+c_1)\lambda(1-)}{\lambda(1-)(1+c_1)^2-2\lambda(1)},\, \mbox{a.s.}, \mbox{ if } \frac{2 \lambda(1)}{\lambda(1-)}> (1+c_1)^2, \label{eq:iopt1}
		\end{align}
		and for any  point $t_i$ of discontinuity of $\lambda$ before $1$, 
		\be \label{eq:iopt2} 
		P_{{t_i}-}=\tilde{V}, \,\mbox{a.s.}.
		\ee
		Moreover, (\ref{eq:iopt1}) and (\ref{eq:iopt2}) are equivalent to
		\begin{align}
		Y_{1-}&=\frac{-\tilde{V}(1+c_1)\lambda(1-)c_{1}-c_0(1+c_1)\lambda(1-)}{\lambda(1-)(\lambda(1-)(1+c_1)^2-2\lambda(1))} +\frac{\tilde{V}-\phi}{\lambda_0}, \, \mbox{a.s.}, \mbox{ if } \frac{2 \lambda(1)}{\lambda(1-)}>(1+c_1)^2,\label{eq:Y1opt} \\
		Y_{t_i}&=\frac{\tilde{V}-\phi}{\lambda(0)}, \,\mbox{a.s.}.\label{eq:Yiopt}
		\end{align}
	Furthermore, the value function of the trader is given by
	$$
              \psi(0,0)+ \sum_{i=1}^{n}\frac{\sigma^2}{2}(1-t_{i})(\lambda(t_i)-\lambda(t_{i-1})) +y^*\frac{(c_1\tilde{V}+c_0)^2}{4\lambda_{1}-2\lambda_{1-}(1+c_1)^2},
        $$
        where 
        \[
            y^*=\left\{\begin{array}{l} 1\mbox{, in the case } \frac{2 \lambda_{1}}{\lambda_{1-}}>(1+c_1)^2\\ 0\mbox{, otherwise.}\end{array}\right.
        \]

	\end{enumerate} 
\end{proposition}

\begin{proof}
We will adapt Back's arguments in \cite{Back92} to our case.  Let $(t_i)_{i=1}^n$ be the times at which $\lambda$  jumps before time $1$.  The arguments of Back can be directly used to show that it is suboptimal for the insider's strategy to jump at times other than $t_i$s. Thus, without loss of generality, $X$ can jump only at $t_i$s or at $1$.

 With the convention that $t_{n+1}=1$ and $t_0=0$ define the function 
      $$
        \psi(t,y)=\sum_{i=0}^n\psi^i(t,y)\chf_{[t_i,t_{i+1})}(t),
      $$
      where
      $$
        \psi^i(t,y)=\frac{1}{2\lambda_i}(\lambda_i (y-Y_{t_i})+ P_{t_i}-\tilde{V})^2+\frac{\sigma^2}{2}\lambda_i(1-t),
      $$
      with $\lambda_i=\lambda(t_{i})$. Direct calculations yield
      \bean
       \psi(1-,Y_{1-})&=&\psi(0,0)+ \int_0^{1-}(P_t -\tilde{V})dY_t\\
       &&+\half \sum_{i=1}^{n}\left[\frac{1}{\lambda_i}(P_{t_i}-\tilde{V})^2-\frac{1}{\lambda_{i-1}}(P_{{t_i}-}-\tilde{V})^2-2(P_{t_i}-\tilde{V})\Delta Y_{t_i}+\sigma^2(1-t_{i})(\lambda_i-\lambda_{i-1})\right].      
      \eean
      
      Using this and admissibility properties of $X$, in particular $d[X,X]_t=(\Delta X_t)^2$,      the insider's optimization problem becomes
\begin{eqnarray*}
\sup_{X\in {\mathcal{A}}}\bbE^{\tilde{v}} [W^{X}_1]&=& \sup_{X\in {\mathcal{A}}}\bbE^{\tilde{v}}  \left[\int_0 ^1 \left(\tilde{V}- P_t\right) dX_t \right] \\
&=& \bbE^{\tilde{v}}\left[\psi(0,0)\right]+\half \sum_{i=1}^{n}\sigma^2(1-t_{i})(\lambda_i-\lambda_{i-1})\\ 
&-&\inf_{X\in {\mathcal{A}}}\bbE^{\tilde{v}}\Bigg[\psi(1-,Y_{1-}) -(\tilde{V}-P_1)\Delta Y_1\\
&&-\half \sum_{i=1}^{n}\left(\frac{1}{\lambda_i}(P_{{t_i}-}-\tilde{V})^2-\lambda_i \left(\Delta Y_{t_i}\right)^2-\frac{1}{\lambda_{i-1}}(P_{{t_i}-}-\tilde{V})^2\right)\Bigg] 
\end{eqnarray*}
where $\bbE^{\tilde{v}}=\bbE^{{v}}$ if $\tilde{V}=V$ and $\bbE$ otherwise. In above the last equality is due to (\ref{a:Xintegrability}). Clearly, $\Delta Y_{t_i}=0$ is optimal for $i\leq n$. Note, however, that the same cannot be said for $\Delta Y_1$.  Indeed, since 
$$ 
    P_1=P_{1-}+\Delta\phi+ \lambda_{n+1}\Delta Y_1=(1+c_1) P_{1-} +c_0+ \lambda_{n+1}\Delta Y_1,
$$ the term $ -(\tilde{V}-P_1)\Delta Y_1$ is minimised at
\[
\Delta Y_1=\frac{{\tilde{V}}-c_0-(1+c_1)P_{1-}}{2\lambda_{n+1}}.
\]
Therefore, with this choice of jumps, the optimisation problem of the insider becomes
\begin{eqnarray}
	\sup_{X\in {\mathcal{A}}}\bbE^{\tilde{v}} [W^{X}_1]	&=& \bbE^{\tilde{v}}\left[\psi(0,0)\right]+\half \sum_{i=1}^{n}\sigma^2(1-t_{i})(\lambda_i-\lambda_{i-1})\nn\\ 
	&-&\inf_{X\in {\mathcal{A}}}\bbE^{\tilde{v}}\Bigg[\frac{1}{2\lambda_n}(P_{1-}-\tilde{V})^2-\frac{1}{4\lambda_{n+1}}((1+c_1)P_{1-}+c_0-\tilde{V})^2 \label{e:minbulk}\\
	&&-\half \sum_{i=1}^{n}\left(\frac{1}{\lambda_i}(P_{{t_i}}-\tilde{V})^2-\frac{1}{\lambda_{i-1}}(P_{{t_i}}-\tilde{V})^2\right)\Bigg] \nn
\end{eqnarray}

Suppose we have $\lambda_i < \lambda_{i-1}$ for some  $i\leq n$. Then by choosing a strategy to make $Y_{{t_i}}$ large, the strategic trader will drive the price $P_{{t_i}}$ arbitrarily far away from $\tilde{V}$, thus, achieving infinite profits. Consequently, $\lambda_i$ are increasing for $i \leq n$ and, therefore, it is optimal to achieve  $P_{t_i}=\tilde{V}$ for all $i\leq n$.

Using the same arguments, we can deduce that $\frac{2 \lambda_{n+1}}{\lambda_n}=\frac{2 \lambda(1)}{\lambda(1-)}\geq (1+c_1)^2$ and that  C2 holds whenever $\frac{2 \lambda(1)}{\lambda(1-)}= (1+c_1)^2$.

On the other hand, if $\frac{2 \lambda_{n+1}}{\lambda_n}> (1+c_1)^2$, the maximizer of
$\frac{1}{4\lambda_{n+1}}((1+c_1)P_{1-}+c_0-\tilde{V})^2-\frac{1}{2\lambda_n}(P_{1-}-\tilde{V})^2$ is given by
\be \label{e:optpt1}
P_{{1}-}= \frac{\tilde{V}((1+c_1)\lambda_n-2\lambda_{n+1})-c_0(1+c_1)\lambda_n}{\lambda_n(1+c_1)^2-2\lambda_{n+1}}=\frac{\tilde{V}((1+c_1)\lambda_{1-}-2\lambda_{1})-c_0(1+c_1)\lambda_{1-}}{\lambda_{1-}(1+c_1)^2-2\lambda_{1}}.
\ee

This yields that the upper bound for the value function of the insider is 
$$
  \psi(0,0)+ \sum_{i=1}^{n}\frac{\sigma^2}{2}(1-t_{i})(\lambda_i-\lambda_{i-1}) +y^*\frac{(c_1\tilde{V}+c_0)^2}{4\lambda_{n+1}-2\lambda_n(1+c_1)^2}.
$$
In above, 
\[
y^*=\left\{\begin{array}{l} 1\mbox{, in the case } \frac{2 \lambda_{n+1}}{\lambda_n}>(1+c_1)^2\\ 0\mbox{, otherwise.}\end{array}\right.
\]


Note that if $t_n>0$, $P_{1-}-\tilde{V}=\lambda(1-)(Y_{1-}-Y_{t_n})=\lambda(1-)(Y_{1-}-
\frac{\tilde{V}-\phi}{\lambda(0)})$. Moreover, if $t_n=0$, we still have $P_{1-}-\tilde{V}=\lambda(1-)(Y_{1-}-
\frac{\tilde{V}-\phi}{\lambda(0)})$ since $\lambda(1-)=\lambda(0)$. Thus, the $Y_{1-}$ corresponding to the equality (\ref{e:optpt1}) in case $ \frac{2 \lambda_{1}}{\lambda_{1-}}>(1+c_1)^2$ is given by
\[
 \lambda(1-)\Big(Y_{1-}-
\frac{\tilde{V}-\phi}{\lambda(0)}\Big)=\frac{-\tilde{V}(1+c_1)\lambda_{1-}c_{1}-c_0(1+c_1)\lambda_{1-}}{\lambda_{1-}(1+c_1)^2-2\lambda_{1}}.
\]

Hence, to conclude the proof, it is enough to show the existence of an admissible strategy for the insider achieving this upper bound. Observe that $Y$ and $P$ are related by
\[
Y_{t_i-}-Y_{t_{i-1}}= \frac{1}{\lambda_{i-1}} \left[P_{{t_i}-}-P_{{t_{i-1}}-}\right], \qquad i =1, \ldots, n+1,
\]
with the convention that $P_{0-}=P_0, \phi(0-)=\phi(0)$.
Thus, the above conditions on $P$ are equivalent to 
\bean
Y_{t_1}=c_1&:=&\frac{\tilde{V}-\phi}{\lambda_0} \\
Y_{t_i}-Y_{t_{i-1}}=c_i&:=&0,\qquad i =2, \ldots, n,\\
Y_{1-}-Y_{t_{n}}=c_{n+1}&:=&\frac{-\tilde{V}(1+c_1)\lambda_{1-}c_{1}-c_0(1+c_1)\lambda_{1-}}{\lambda_{1-}(\lambda_{1-}(1+c_1)^2-2\lambda_{1})}.
\eean
Direct calculations yield (\ref{eq:Y1opt}) and (\ref{eq:Yiopt}).

A strategy that drives $Y$ to the above levels is given by
\[
\alpha_t= \sigma^2 \frac{Y_t-Y_{t_{i-1}}-c_i}{t_i-t}, \qquad t \in [t_{i-1}, t_i), \, i=1, \ldots, n+1.
\]
Indeed, this strategy makes $Y_{-}$ a Brownian bridge over the interval $[t_{i-1}, t_i]$ from $Y_{t_{i-1}}$ to $Y_{t_{i-1}}+c_i$. In particular, $Y_{-}$ is Gaussian conditioned on $\tilde{V}=v$. Since $\lambda$ as well as $\phi$ is deterministic on $[0,1)$, $P$ is a Gaussian process on $[0,1)$, too. This implies $\bbE^{\tilde{v}}[(P_t-\phi(t))^2]$ is a finite continuous function of $t$ on $[0,1)$ and, therefore, the above strategy is admissible. 
\end{proof}

\begin{remark}\label{r:DMM1}
	In the setting of Proposition \ref{p:ioptimality}, if the strategic (informed or otherwise) trader acts optimally then
	\be \label{e:lossDMM} 
	\frac{1}{2\lambda(1-)}(P_{1-}-\tilde{V})^2-(\tilde{V}-P_1)\Delta Y_1 \leq 0.
	\ee
	Indeed, (\ref{e:dY1opt}) implies $\tilde{V}-P_1=\lambda(1)\Delta Y_1$. Therefore,
	\[
	\frac{1}{2\lambda(1-)}(P_{1-}-\tilde{V})^2-(\tilde{V}-P_1)\Delta Y_1 =\frac{1}{2\lambda(1-)}(P_{1-}-\tilde{V})^2-\frac{((1+c_1)P_{1-}+c_0-\tilde{V})^2}{4\lambda(1)}.
	\]
	However, the right hand side above is minimised by the strategic trader in view of (\ref{e:minbulk}). Note that the minimum is at most $0$ since she  can trade so that $P_{1-}=\tilde{V}$. 
\end{remark}
\begin{remark} \label{r:fictitiousgain}
	Consider the gains process $\hat{G}$ associated to a fictitious market maker that holds all the orders submitted by the strategic trader  and the noise traders and trades at the market quotes. This gains process corresponds to the sum of all the gains of the market makers. Observe that
	\[
	\hat{G}_1=-\int_0^{1-}Y_s dP_s+ Y_{1-}(P_{1-}-V)+ \Delta Y_1(P_1-V)=-\int_0^1 (V-P_t)dY_t +\sigma^2 \int_0^1 \lambda(t)dt.
	\]
	Thus,  when the strategic trader is informed, 
	\[
	\begin{split}
\hat{G}_1&=\sigma^2 \int_0^1 \lambda(t)dt+\psi(1-,Y_{1-})-\psi(0,0)-(V-P_1)\Delta Y_1  \\
&-\half \sum_{i=1}^{n}\left[\frac{1}{\lambda_i}(P_{t_i}-{V})^2-\frac{1}{\lambda_{i-1}}(P_{{t_i}-}-{V})^2-2(P_{t_i}-{V})\Delta Y_{t_i}\right]\\
&-\half\sigma^2\int_0^{1-}(1-t)d\lambda(t)\\
&=\frac{1}{2\lambda(1-)}(P_{1-}-V)^2-(V-P_1)\Delta Y_1 -\frac{1}{2\lambda(0)}(P_0-V)^2 \\
& +\half\sigma^2 \int_0^1 \lambda(t)dt\\
&-\half \sum_{i=1}^{n}\left[\frac{1}{\lambda_i}(P_{t_i}-{V})^2-\frac{1}{\lambda_{i-1}}(P_{{t_i}-}-{V})^2-2(P_{t_i}-{V})\Delta Y_{t_i}\right]
	\end{split}
	\]
	Since the insider's optimal strategy has no jumps before time $1$ and $P_{t_i}=V$ by optimality considerations, the last term vanishes.  Thus, 
	\be \label{e:BGIf}
	\hat{G}_1= \frac{1}{2\lambda(1-)}(P_{1-}-V)^2-(V-P_1)\Delta Y_1 -\frac{1}{2\lambda(0)}(P_0-V)^2  +\half\sigma^2 \int_0^1 \lambda(t)dt.
	\ee
	
	If the strategic trader is uninformed, 
	\[
		\hat{G}_1=-\int_0^{1-}Y_s dP_s+ Y_{1-}(P_{1-}-\mu)+ \Delta Y_1(P_1-\mu)+ Y_1(\mu-V).
	\]
Thus,  similar considerations yield
	\be \label{e:BGuf}
\hat{G}_1= \frac{(P_{1-}-\mu)^2}{2\lambda(1-)}-(\mu-P_1)\Delta Y_1 -\frac{(P_0-\mu)^2 }{2\lambda(0)}+ Y_1(\mu-V)
 +\half\sigma^2 \int_0^1 \lambda(t)dt.
\ee
\end{remark}
}


\end{document}